%% file: main.tex
\documentclass[11pt,a4paper]{article}
\usepackage{svg}
\usepackage{subcaption}
\usepackage[
backend=biber,
style=alphabetic,
sorting=ynt
]{biblatex}
\usepackage[T1]{fontenc}
\usepackage[utf8]{inputenc}
\usepackage{authblk}
\usepackage{titling}
\usepackage[left=1in, right=1in, top=1in, bottom=1in]{geometry}
\usepackage{mathtools}
\usepackage{dsfont}
\usepackage{amsthm}
\usepackage{thmtools}
\usepackage{thm-restate}
\usepackage{hyperref}
\usepackage{cleveref}
\usepackage{physics}
\usepackage{mymacros}
\usepackage{amssymb}
\usepackage{amsmath}
\PassOptionsToPackage{dvipdfmx}{graphicx}
\usepackage{dcolumn}
\usepackage{hyperref}

\declaretheorem[numberwithin = section]{theorem}
\declaretheorem[sibling = theorem]{lemma}
\declaretheorem[sibling = theorem, style = definition]{definition}
\declaretheorem[sibling = theorem]{proposition}
\declaretheorem[numbered = no, style = remark]{remark}
\declaretheorem[sibling = theorem]{corollary}

\bibliography{paperpile}

\begin{document}


\author{Siddharth Vadnerkar \thanks{email: \href{svadnerkar@ucdavis.edu}{svadnerkar@ucdavis.edu}, site: \href{https://sites.google.com/view/siddharthvadnerkar}{https://sites.google.com/view/siddharthvadnerkar}}}
\affil{Department of Physics, University of California, Davis}

\title{Superselection sectors in the 3d Toric Code}




\date{\today}

\maketitle
\begin{abstract} 
We rigorously define superselection sectors in the 3d (spatial dimensions) Toric Code Model on the infinite lattice $\mathbb{Z}^3$. We begin by constructing automorphisms that correspond to infinite flux strings, a phenomenon that's only possible in open manifolds. We then classify all ground state superselection sectors containing infinite flux strings, and find a rich structure that depends on the geometry and number of strings in the configuration. In particular, for a single infinite flux string configuration to be a ground state, it must be monotonic. For configurations containing multiple infinite flux strings, we define "infinity directions" and use that to establish a necessary and sufficient condition for a state to be in a ground state superselection sector. Notably, we also find that if a state contains more than 3 infinite flux strings, then it is not in a ground state superselection sector.

\end{abstract}

\tableofcontents

\nocite{*}

\input{Main/Intro.tex}

\input{Main/3D_Toric_Code.tex}
\input{Main/GS_purely_charged.tex}
\input{Main/GS_on_1str.tex}
\input{Main/GS_on_2str.tex}
\input{Main/GS_on_3str.tex}
\input{Main/Translations.tex}
\input{Main/Discussion.tex}

\appendix

\input{Main/purity.tex}
\input{Main/Lattice.tex}

\section{Acknowledgements}
We would like to thank Bruno Nachtergaele for the initial idea for the paper and insightful discussions. We would also like to thank Alex Bols for a detailed sketch of the proof of purity of the frustration free ground state of the 3d TC. We were supported by the NSF grant DMS-2108390.
\printbibliography

\end{document}

%% file: Main/Intro.tex
\section{Introduction}
In recent years, the subject of topological phases has exploded in popularity. A number of works have been written exploring various facets of these phases. Topological phases exhibit robust ground state degeneracy, one that is independent of the microscopics of the system. They also have topological excitations like anyons and flux strings. These excitations can be used for topological quantum computation \cite{Nayak2008-ef}. Various systems like the Fractional Quantum Hall state and the spin-liquid states exhibit the highlighting features of topological phases and have been discovered in nature \cite{Nakamura2020-vb, Stormer1999-bk}.\\

In particular, Kitaev's Toric Code \cite{Kitaev2003-qr} is an interesting 2d toy model that exhibits various features of topological phases. This model has been instrumental in understanding the nature of topological phases, as one can explicitly compute many quantities in it such as explicit string operators that detect and move topological charges, or the ground state degeneracy on a discrete manifold.\\

Kitaev originally studied this model on a discrete torus, obtaining a ground state degeneracy of 4. He then extended this analysis to find the ground state degeneracy on (the discretization of) a closed manifold of arbitrary genus. It was shown later that the ground state superselection sectors also exist when this model is placed on an open manifold such as a discretization $\mathbb{Z}^2$ of the 2d plane $\mathbb{R}^2$ \cite{Cha2020-rz}.\\

Several generalisations of this model exist. In particular, the 3d Toric Code is a model in 3+1 spacetime dimensions. It is of theoretical interest with respect to error-correcting codes \cite{Kulkarni2019-gw}. And a variety of interesting topological features of this model have been demonstrated \cite{Kong2020-an, Castelnovo2008-yg}. The ground state degeneracy of this model has already been shown on closed 3d manifolds with an arbitrary genus (see for example \cite{Von_Keyserlingk2013-zv}). In this work, we analyse the ground state superselection sector structure of the 3d Toric Code on $\mathbb{Z}^3$, which is a discretization of the open manifold $\mathbb{R}^3$. We will now omit the usage of "the discretization of" and assume it implicitly, as the model is only defined on discrete lattices. Understanding this structure is important as it results in consistency conditions for when fusion and braiding of string-like objects can occur in non-compact manifolds.\\

Our reason for studying the 3d Toric Code on an open manifold is that it exhibits certain structures that are not found in the closed manifolds: infinite flux strings. In a closed manifold, there are 2 types of excitations in the 3d Toric Code - flux strings and charges. While charges are topological particle-like excitations, flux strings are topological string-like excitations. The flux strings must be closed and finite in energy. However in open manifolds they are not required to be closed. We call such excitations infinite flux strings, and they are "infinite energy" excitations. Due to this, they cannot be physically obtained from the ground state using local operators. However they are still a physically relevant topic of study as they may be interpreted as boundary conditions of the model.\\

In studying these excitations, we find a rich structure. We find that some configurations of infinite string excitations are stable in the sense that their energy cannot be decreased arbitrarily, while some others are unstable. The configurations that are stable belong to ground state superselection sectors, which we precisely define in section \ref{sec:recap}. The aim of this paper will be to classify all ground state superselection sectors of the 3d Toric Code.\\

The layout of this paper is as follows: In section \ref{sec:recap} we briefly recall the 3d Toric Code on $\mathbb{Z}^3$ and construct automorphisms corresponding to infinite flux strings. We then present a summary of the main results of this paper. In section \ref{sec:pure-charge} we sketch the construction of charged sectors in the 3d Toric Code. In sections \ref{sec:1str},\ref{sec:2str},\ref{sec:3str} we first construct infinite flux string sectors. Then we tackle the question of the necessary and sufficient conditions for a ground state configuration. We then introduce "Infinity directions" and use that as a basis for proving these conditions. Finally we classify all possible ground state sectors for any configuration with arbitrary number of flux strings.

%% file: Main/3D_Toric_Code.tex
\section{3d Toric Code model, superselection sectors, main results}
\label{sec:recap}


\subsection{3d Toric Code model}
We begin by describing the 3d Toric Code in the $C^*$ algebraic framework \cite{Alicki2007-pk}. Consider the lattice $\Gamma = \mathbb{Z}^3$ which is an oriented cell complex. Denote the set of vertices, oriented edges, oriented faces in this cell complex as $\mathcal{V}(\Gamma), \mathcal{E}(\Gamma),\mathcal{F}(\Gamma)$ respectively. We fix an orientation of the edges as follows: all edges are pointing in the positive $x,y,z$ direction. Let $\partial$ denote the standard boundary map of this cell complex. We place a qubit on each edge $e \in \mathcal{E}(\Gamma)$, so there is an edge Hilbert space $\hilb_e = \mathbb{C}^2$ with observables $\mathcal{B}(\hilb_e) = M_2(\mathbb{C})$. Let $\Lambda_f$ denote the set of all finite subsets of $\mathcal{E}(\Gamma)$. Then $\hilb_\Lambda = \otimes_{e \in \Lambda} \hilb_e$ for $\Lambda \in \Lambda_f$. We can define an algebra on this space as $\cstar{\Lambda} = \mathcal{B}(\hilb_\Lambda)$. \\

Let $\Lambda_1, \Lambda_2 \in \Lambda_f$ and $\cstar{\Lambda_1}, \cstar{\Lambda_2}$ be algebras such that $\Lambda_1 \subset \Lambda_2$. $\iota: \cstar{\Lambda_1} \xhookrightarrow{} \cstar{\Lambda_2}$ is the embedding such that $$\iota (\cstar{\Lambda_1}) = \cdots \mathds{1} \otimes \mathds{1} \otimes \cstar{\Lambda_1} \otimes \mathds{1} \otimes \mathds{1} \cdots $$ where $\mathds{1}$ is on all edges $e\in \Lambda_2 \setminus \Lambda_1$. We can now define $$\cstar{loc} := \bigcup_{\Lambda \in \Lambda_f} \cstar{\Lambda} \quad \quad \cstar{} := \overline{\cstar{loc}}$$ where $\cstar{}$ is the $C^*$  algebra for the 3d Toric Code, the algebra of quasi-local operators in the 3D Toric Code.\\

The Hamiltonian for the 3D Toric Code inside a finite region $\Lambda$ is given by:
\begin{align}
\label{eqn:Hamiltonian}
    &H_\Lambda = \sum_{v \in {\Lambda}} (\mathds{1} - A_v) + \sum_{f \in \Lambda} (\mathds{1} - B_f)\\
    &A_v := \prod_{ v \in e} \sigma_e^x \qquad B_f := \prod_{e \in f} \sigma_e^z
\end{align}
The operator $A_v$ acts on the 6 edges that surround $v$. There are 3 different $B_f$ operators corresponding to the orientation of the face $f$. Each $B_f$ operators acts on the 4 edges that make up the face $f$. $A_v$ and $B_f$ operators commute. They both square to $\mathds{1}$, thus having eigenvalues $\pm 1$.\\

Since the interactions $A_v, B_f$ are translation invariant, there exists an action $\alpha_t$ of $\mathbb{R}$ on $\cstar{}$ describing the dynamics of the system, as well as a derivation $\delta$ that is the generator of the dynamics \cite{Bratteli2012-gd}. \\

Let us recall some properties of this model.  Refer to \cite{Alicki2007-pk, Cha2018-ke, Bachmann2016-qx, Bachmann2023-qv} for a full treatment in the case of the 2d Toric Code, which proceeds similarly. The Hamiltonian \ref{eqn:Hamiltonian} is not convergent in norm, but still generates dynamics through a derivation, $\delta({O}) := \lim_{\Lambda \rightarrow \mathcal{E}(\Gamma)} i[H_\Lambda, {O}]$ for all ${O} \in \cstar{loc}$. $\delta$ can be extended to a densely defined unbounded *-derivation on $\cstar{}$.\\

In the infinite lattice setting, states are described by positive linear functionals $\omega: \cstar{} \rightarrow \mathbb{C}$. Given $(\omega, \cstar{})$, we can construct a unique GNS triple $(\pi, \hilb, \Omega)$ up to unitary equivalence, such that $\omega({O}) = \langle \Omega, \pi ({O}) \Omega\rangle$. Here $O \in \cstar{}$, $\pi$ is a representation of $\cstar{}$ on a Hilbert space $\hilb$, and $\Omega \in \hilb$ is a cyclic vector.\\

\begin{definition}[\textit{ground state}]
\label{def:GS}
    A state $\omega$ is a ground state if for all $O \in \cstar{}$ we have $$ -i \omega(O^\dagger\delta(O)) \ge 0$$ The energy $E$ of a finite region $\Lambda$ is given by $E_\Lambda := \omega(H_\Lambda)$.
\end{definition}

\begin{remark}
This definition of a ground state is equivalent to saying $\omega(O^\dagger H_\Lambda O)/\omega(O^\dagger O)\ge \omega(H_\Lambda)$ for all finite regions $\Lambda$. This implies if the energy of a state $\omega$ inside any finite region $\Lambda$ is the lowest possible energy, then $\omega$ is a ground state.
\end{remark}

Choose a state $\omega$ such that $\omega(A_v) = \omega(B_f) = 1$ for all $v,f$. Call its GNS triple $(\pi_\omega, \hilb_\omega, \Omega_\omega)$. $\Omega_\omega$ then satisfies $\pi_\omega(A_v)\Omega_\omega = \pi_\omega(B_f) \Omega_\omega = \Omega_\omega$. We have the following fact for $\omega$, which will be proved in the appendix \ref{app:purity}:

\begin{restatable}{theorem}{puregs}
    $\omega$ is the unique pure frustration free translation invariant ground state of the 3dTC. 
    \label{thm:puregs}
\end{restatable}

\subsection{Constructing excited states}
\label{sec:constexcite}
Let $\partial_0 e_i$ ($\partial_1 e_i$) denote the start (end) vertex of $e_i$, and let the boundary map $\partial e_i := \partial_1 e_i - \partial_0 e_i$.

\begin{definition}[\textit{Finite path}]
A finite set $\gamma := \{e_i\}_{i=0}^{l-1} \subset \mathcal{E}(\Gamma)$ is a finite path on the lattice if it satisfies $\partial_1 e_i = \partial_0 e_{i+1}$ for $0\le i < l$ and does not self-intersect (there does not exist $\gamma' \subset \gamma$ such that $\sum_{e' \in \gamma'} \partial e' = 0$). We call $|\gamma|$ the length of the finite path. $\gamma$ is a finite \textit{open} path if $\gamma$ is a finite path and satisfies $\sum_{i=0}^{l-1} \partial e_i = \partial_1 e_{l-1} - \partial_0 e_0 \neq 0$. We denote the start of the path as $\partial_1 \gamma = \partial_1 e_{l-1}$ and end of the path as $\partial_0 \gamma = \partial_0 e_0$. Together we refer to the start and end of $\gamma$ collectively as $\partial \gamma$. $\gamma_c$ is a finite \textit{closed} path if $\gamma_c$ is a finite path and satisfies $\sum_{i=0}^{d-1} \partial e_i= 0$ and thus has no start or end. A finite path $\dg$ is called \textit{trivial} if $\gamma = \emptyset$.
\end{definition}

Any finite self-intersecting path can be decomposed into an open path and finitely many closed paths. So we consider only paths that don't self intersect.\\

One can build charged states by considering the charge operators $$F_\gamma = \prod_{e\in \gamma} \sigma^z_e\in \cstar{}$$ where $\gamma$ is a finite open path on the lattice. Consider an inner automorphism $\alpha_\gamma (O):= F_\gamma O F_\gamma$. The representation $\pi_\omega \circ \alpha_\gamma$ then gives an excited state $\omega_\gamma := \langle \Omega_\omega, \pi_\omega \circ \alpha_\gamma (\cdot) \Omega_\omega\rangle $ in the same GNS Hilbert space. $A_v$ commutes with $F_\gamma$ on all $v \notin {\partial \gamma}$. If $v \in {\partial \gamma}$, we have $A_v F_\gamma = - F_\gamma A_v$ implying $\omega_\gamma(A_v) = -1$. So the state has 2 excitations, one at each endpoint of $\gamma$. These excitations are called charges.\\

One can similarly introduce another kind of excitation called a flux string. We first define $\dG = \mathbb{Z}^3$ as the dual lattice to $\Gamma$ (indeed it is a cell complex dual to $\Gamma$). Denote the set of vertices, edges, faces in $\dG$ as $\mathcal{V}(\dG), \mathcal{E}(\dG), \mathcal{F}(\dG)$. Of particular relevance is the fact that each edge $e \in \mathcal{E}(\Gamma)$ has a unique dual face $f \in \mathcal{F}(\dG)$ and similarly each face $f\in \mathcal{F}(\Gamma)$ has a unique dual edge $e \in \mathcal{E}(\dG)$. Flux loops can be defined as dual-paths on the lattice, or equivalently as paths on the dual lattice. We will adopt the latter nomenclature. To avoid confusion regarding the lattice a path or surface belongs to, we will use an overline ($\overline{\cdot}$) when talking about paths and surfaces on a dual lattice. \\


Let the boundary map on $f \in \mathcal{F}(\Gamma)$ be given by $\partial f = \sum_{e \in f} e$.

\begin{definition}[\textit{Surface}]
     A finite set $S = \{f_i\} \subset \mathcal{F}(\Gamma)$ is a finite surface if it satisfies $\sum_i \partial f_i = \sum_{e \in \dg_c} e$ for some finite closed non self-intersecting path $\gamma_c$, and does not self-intersect (there does not exist $S' \subset S$ such that $\sum_{f' \in S'} \partial f' = 0$). $\gamma_c$ is called the boundary of surface $S$ and denoted as $\partial S$. A dual surface is similarly defined on $\dG$. A surface $S$ is called an \textit{open} surface if its boundary $\dg_c$ is a non-trivial path. It is called a \textit{closed} surface if its boundary $\dg_c$ is a trivial path.
\end{definition}

Any self-intersecting surface can be decomposed into a open surface and finitely many closed surfaces. So we consider only surfaces that don't self-intersect. Similarly for surfaces with a self-intersecting boundary.\\

Consider the flux string operator $$F_\dS = \prod_{e \perp \dS} \sigma^x_e \in \cstar{}$$ where by $e \perp \dS$ we mean $e $ is dual to a given $\overline{f} \in \dS$. Analogous to the case of charged excitations, consider an inner automorphism $\alpha_\dg (O) := F_\dS (O) F_\dS$  such that $\dg = \partial \dS$. $\pi_\omega \circ \alpha_\dg$ then gives us then gives us an excited state $\omega_{\dg}:= \langle \Omega_\omega, \pi_\omega \circ \alpha_\dg (\cdot) \Omega_\omega\rangle$. We have $B_f F_\dS = - F_\dS B_f$ if $f$ is dual to some $\duale \in \dg$. They commute otherwise. This implies for such $f$, $\omega_\dg(B_f) = -1$. So $F_\dS$ produce excitations along $\dg$. So the energy of the excited state is proportional to the size of the boundary $\dg$.

\subsection{Superselection sectors}
\begin{definition}[\textit{Equivalence}]
    Given two representations $\pi_1, \pi_2$ of $\cstar{}$ on $\hilb_1, \hilb_2$ respectively, $\pi_1, \pi_2$ are equivalent if there exists a bounded linear unitary map $U:\hilb_1 \rightarrow \hilb_2$ such that $\pi_1 = U \pi_2 U^\dagger$. We denote [$\pi$] as the equivalence class of $\pi$. We say the states $\omega_1, \omega_2$ are equivalent if their corresponding representations $\pi_1, \pi_2$ are equivalent.
\end{definition}

In general, there are many equivalence classes of representations. A lot of such representations are physically uninteresting due to a variety of reasons (for example, the energy may be unbounded \cite{Naaijkens2010-aq}). To restrict to a physically interesting class of representations, we additionally employ a \textit{superselection criterion}. This criterion tells us which representations we should select.\\

Doplicher-Haag-Roberts (DHR) analysis in algebraic quantum field theory shows that with a physically motivated superselection criterion one can recover all the physically relevant properties like braiding and fusion of charges \cite{Doplicher1971-jd, Doplicher1974-hb}. A similar analysis has been done for the 2d Quantum Double models \cite{Fiedler2015-na, Naaijkens2010-aq}.\\

We choose the following superselection criterion for selecting representations. Consider a continuous manifold $\mathbb{R}^3$ and a differentiable curve $L$ in it. Let plane $P$ be the plane perpendicular to the tangent vector of $L$ at point $p$. Let $\vec{f}(p)$ be a vector function that gives a "framing" vector $\vec{f}(p)$ lying inside plane $P$ for every point $p \in L$ and varies smoothly from a point $p \in L$ to any other point $q \in L$. Figure \ref{fig:wedge} visualises this construction.\\

We adopt the vector notation for points on the plane $P$. Let the point $\vec{p}$ be the origin of plane $P$. We can then define $\Delta_{p, \vec{f}(p), \theta}$ as an infinite triangular section of the plane $P$, with the starting vertex $\vec{p} \in L$ and consisting of all points $\vec{k} \in P$ such that $0 \le \frac{(\vec{k} - \vec{p}) \cdot \vec{f}(p)} {|(\vec{k} - \vec{p}) \cdot \vec{f}(p)|} \le \cos(\theta/2)$. We impose $\theta < \pi$ to ensure the section is triangular (refer to \ref{fig:wedge}). We call an infinite wedge as $W_{L, \vec{f}(L), \theta} = \mathcal{E}(\Gamma) \bigcap (\bigcup_{p \in L} \Delta_{p, \vec{f}(p), \theta})$, i.e, the set of all edges $e \in \mathcal{E}(\Gamma)$ that lie inside $\bigcup_{p \in L} \Delta_{p, \vec{f}(p), \theta}$.\\

Now let $W = \bigcup_i W_i$ be the union of a finite number of non-intersecting wedges $W_i$ defined as above (keeping the variables implicit), and let $W^c$ be the complementary region to $W$. Let an irreducible representation $\pi'$ be such that there exists $W$ with  $$[\pi_\omega \restriction \cstar{W^c}] = [\pi' \restriction \cstar{W^c}]$$ Then we select all representations $\pi$ which satisfy $[\pi] = [\pi']$. Here $\pi_\omega$ is the irreducible representation of the unique frustration-free ground state. \\

\begin{figure}
    \centering
    \includegraphics[width = 0.4 \textwidth]{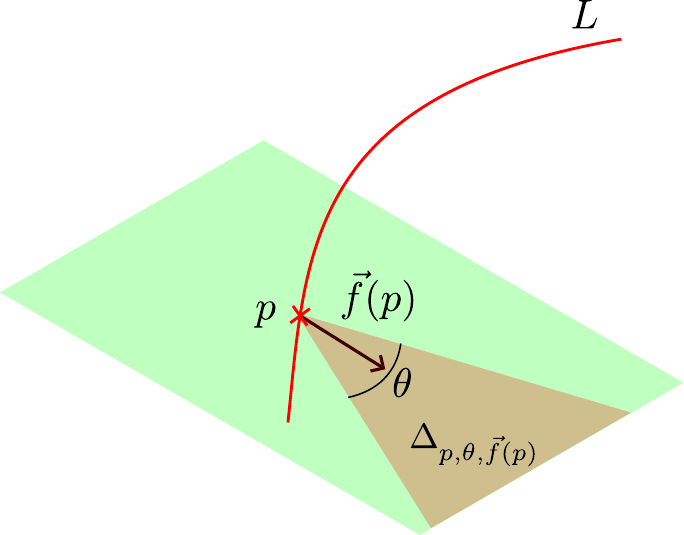}
    \caption{Defining $\Delta_{p,\vec{f}(p), \theta}$}
    \label{fig:wedge}
\end{figure}


\begin{remark}
    The superselection criterion is chosen to allow for representations with finitely many infinite flux strings (to be defined below). The criterion also allows for representations with finitely many charges (to be defined below).
\end{remark}

We call the GNS Hilbert spaces of inequivalent representations $\pi$ satisfying the superselection criterion as superselection sectors. We will refer to superselection sectors simply as sectors. We can construct sectors with different charges by considering the automorphism\begin{align}
    \alpha_v ({O}) := \lim_{n\rightarrow \infty} F_{\gamma_{v;n}} {O} F_{\gamma_{v;n}} \qquad {O} \in \cstar{}
\end{align}
where $\gamma_{v;n}$ is a path starting at $v \in \mathcal{V}(\Gamma)$ and stretching straight down to $v - n \hat{z} \in \mathcal{V}(\Gamma)$. The representation $\pi_v := \pi_\omega \circ \alpha_v$ defines a representation of $\cstar{}$ on to a new sector we call $\hilb_\epsilon$. Call the corresponding state $\omega_v$. It is a ground state, although of the sector $\hilb_\epsilon$. A general charged state having charges at $v_1,\cdots v_N$ would be given by $\omega_{v_1,\cdots, v_N} := \langle \Omega_\omega, \pi_{v_1, \cdots , v_N} (\cdot) \Omega_\omega \rangle$, where $\pi_{v_1, \cdots , v_N} := \pi_\omega \circ \alpha_{v_1} \circ \cdots \circ \alpha_{v_N}$.\\

We can also construct an infinite flux string 
 state through a similar construction. Consider an automorphism -
\begin{align}
    \alpha_\dg ({O}):= \lim_{n\rightarrow \infty} F_{\dS_{\dg_n}} {O} F_{\dS_{\dg_n}} \qquad {O} \in \cstar{}
\end{align}
where $\dS_{\dg_n}$ is a surface such that a finite path $\dg_n \subset \partial \dS_{\dg_n}$ and $\dg = \lim_{n\rightarrow \infty} {\dg_n}$ is an infinite path in $\overline{\Gamma}$ (defined below). $\pi_\dg := \pi_\omega \circ \alpha_\dg$ defines a representation of $\cstar{}$ to a new sector $\hilb_{\dg}$ with $\dg$ as an infinite flux string. Call the corresponding state $\omega_\dg$. Again, composing multiple such automorphisms gives us a state $\omega_{\dg_1,\cdots, \dg_ N}:= \langle \Omega_\omega,\pi_{\dg_1,\cdots, \dg_N}(\cdot) \Omega_\omega \rangle$ where $\pi_{\dg_1,\cdots, \dg_N} := \pi_\omega \circ \alpha_{\dg_1} \circ \cdots \circ \alpha_{\dg_N}$.\\

\begin{restatable}{theorem}{purity}
    All ground states of the form $\omega_{v_1, \cdots, v_n; \dg_1,\cdots ,\dg_m}$ (including $\omega$, the unique translation invariant ground state) are pure.
    \label{thm:purity}
\end{restatable}
We will prove this theorem in appendix \ref{app:purity}. Purity of states is important since it implies the corresponding GNS representation is irreducible (\cite{Naaijkens2013-ji}, theorem 2.5.14), and one can then talk about equivalence of irreducible representations to be selected by the superselection criterion. A superposition of states in different sectors is necessarily a mixed state.

\subsection{Main results}
We now provide a summary of the main results in the paper.\\

A state is purely charged if it only has flux excitations with a finite boundary. A purely charged sector is the sector containing only purely charged states.

\begin{restatable}{theorem}{purelycharged}
\label{thm:charged-sectors}
        There are only 2 purely charged ground state sectors, given by $\hilb_\omega, \hilb_\epsilon$.

\end{restatable}

We will call a state charged if it lies in the $\hilb_\epsilon$ sector, and uncharged otherwise.

\begin{definition}[\textit{Infinite path}]

\label{def:infpath}
    An infinite path in $\dG$ is a function $\dg: \mathbb{Z} \rightarrow \edge(\dG)$ such that for any $a,b \in \mathbb{Z}$, $\dg[a,b] := \{\dg(t)\}_{t=a}^{t=b}$ is a finite path. We say $\duale \in \dg$ if there exists $t \in \mathbb{Z}$ such that $\dg(t) = \duale$.
\end{definition}

\begin{remark}
    The definition of a finite path necessarily forces an infinite path to be non self-intersecting.
\end{remark}



\begin{remark}
An infinite flux string is mathematically the automorphism $\alpha_\dg$ on an infinite path $\dg$ in $\overline{\Gamma}$.    
\end{remark}

Each oriented edge can point in the $(r,\sigma)$ direction, where $r\in \{x,y,z\}$ and $\sigma \in \{\pm\}$. Let $ \edge^{r,\sigma}$ denote the set of all edges pointing in the $(r,\sigma)$ direction. Then we have $\mathcal{E}(\overline{\Gamma}) =\cup_{r ,\sigma} \edge^{r,\sigma}$.

\begin{definition}[\textit{Monotonicity}]
\label{def:monopath}
    An infinite path $\dg$ is monotonic if for each $r$ there exists $\sigma_r \in \{\pm\}$ such that $\dg \subseteq \cup_{r} \edge^{r,\sigma_r}$.  
\end{definition}
\begin{remark}
An example of a monotonic and non-monotonic path are shown in figure \ref{fig:mono_vs_non}.
\end{remark}

\begin{figure}[t!]

        \centering
    \begin{subfigure}[t]{0.5\textwidth}
        \centering
\includegraphics[ width=0.5\textwidth]{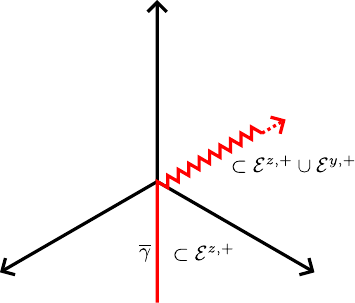}
    \end{subfigure}%
    ~ 
    \begin{subfigure}[t]{0.5\textwidth}
        \centering
\includegraphics[ width=0.5\textwidth]{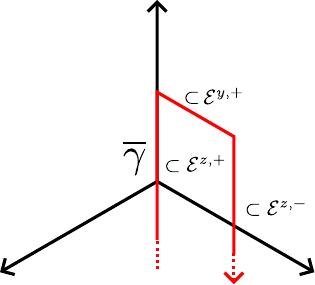}
    \end{subfigure}
    \caption{Examples of monotonic and non-monotonic paths}

    \label{fig:mono_vs_non}
\end{figure}


In what follows, we implicitly assume $\dg$ is an infinite path, unless stated otherwise.

\begin{restatable}{theorem}{monotonicGS}
\label{thm:mono-GS}
        A state $\omega_\dg$ is a ground state iff $\dg$ is a monotonic path in $\dG$.
\end{restatable}

Performing local operations on an infinite flux string will not change its sector. Thus it is not the microscopic behaviour of $\dg$ that matters to understand the sector it lies in, but the behaviour of $\dg$ in the macroscopic or infinite limit. To this end, we define a few key terms.

\begin{definition}[\textit{Infinity direction}]
\label{def:InfDir}
    Pick $n_0 \in \mathbb{Z}$ and a path $\dg$. $$D_+(\dg) := \biggl\{(i,\sigma)\bigg| i \in \{x,y,z\}, \sigma \in \{\pm\},\#\{n \in \mathbb{Z}| n > n_0, \dg(n) \subseteq \edge^{i,\sigma}\} = \infty \biggr\}$$ $D_+(\dg)$ is called the set of positive infinity directions. Similarly, $$D_-(\dg) := \biggl\{(j,\tau)\bigg| j \in \{x,y,z\}, \tau \in \{\pm\},\#\{n \in \mathbb{Z}| n < n_0, -\dg(n) \subseteq \edge^{j,\tau}\} = \infty \biggr\}$$ where $-\dg(n)$ is the edge $\dg(n)$ but pointing in the opposite direction. $D_-(\dg)$ is called the set of negative infinity directions. $D(\dg) := D_+(\dg) \cup D_-(\dg)$ is the set of infinity directions.
\end{definition}

\begin{remark}
    The definition of infinity directions is independent of the choice of $n_0$, as can be checked by relabelling $n \mapsto n+t$.
\end{remark}

\begin{definition}[\textit{Path equivalence}]
    Let $\dg_1, \dg_2$ be two paths. Define 
\begin{align*}
    r = min\{t | \dg_1 (t) \notin \dg_2\}\\
    s = max\{t | \dg_1 (t) \notin \dg_2\}
\end{align*}
Then $\dg_1, \dg_2$ are path equivalent if $-\infty < r,s < \infty$
\end{definition}
An example of two path equivalent configurations is given in figure \ref{fig:path_eqv}.

\begin{figure}[t!]
    \centering
    \begin{subfigure}[t]{0.45\textwidth}
    \centering
    \includegraphics[width=0.45\textwidth]{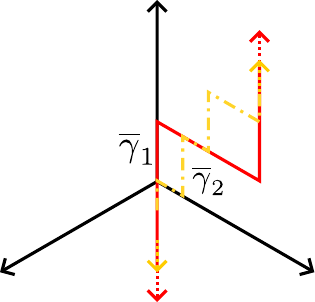}
    \caption{An example of two path equivalent configurations.}
    \label{fig:path_eqv}        
    \end{subfigure}
    \begin{subfigure}[t]{0.45\textwidth}
    \centering
    \includegraphics[width=0.45 \textwidth]{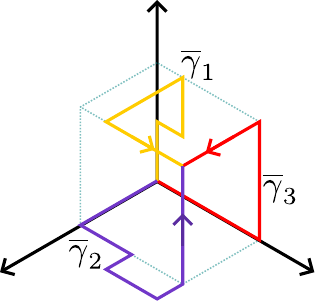}
    \caption{An example of different monotonic and non-monotonic paths between two endpoints of $\dL$. $\dg_2$ is non-monotonic while $\dg_1, \dg_3$ are monotonic.}
    \label{fig:mono_takes_least}      
    \end{subfigure}
    \caption{}

\end{figure}




\begin{restatable}{theorem}{gss}
    \label{thm:GSS}
Let $\{\dg_n\}_{n=1}^{N}$ be a set of monotonic paths. Then $\omega_{\dg_1, \cdots, \dg_N}$ lies in a ground state sector iff $\bigcap_{i=1}^N D(\dg_i) = \emptyset$.
\end{restatable}






\begin{restatable}{theorem}{onesectorclass}
\label{thm:onesectorclass}
    All {\em inequivalent} ground states with 1 infinite flux string are labelled by $(g \in \mathbb{Z}_2,\dg\in \mathcal{M})$ where $g$ indicates if the ground state is charged or uncharged, $\mathcal{M}$ is the set of all path inequivalent monotonic paths. 
\end{restatable}

\begin{definition}[\textit{Half-infinite path}]
\label{def:halfpath}
    A positive half-infinite path on $\dG$ is a function $\dg_v: \mathbb{Z}_+ \rightarrow \mathcal{E} (\dG)$ such that for any $[a,b] \in \mathbb{Z}_+$, $\dg[a,b]:= \{\dg(t)\}_{t=a}^{t=b}$ is a finite path. We call $\partial_0 \dg_v = v \in \mathcal{V}(\dG)$ as the starting vertex of $\dg_v$. A negative half-infinite path is similarly defined as a function $\dg_v: \mathbb{Z}_- \rightarrow \mathcal{E}(\dG)$ that has $\partial_1 \dg_v = v$ as its ending vertex.
\end{definition}

    Let $\dg^{(r,\sigma)}_v$ with $r\in \{x,y,z\},\sigma \in \{\pm \}$ be a positive/negative half-infinite path that has its start/end point as $v$ and $\dg^{(r,\sigma)}_v(t) \in \mathcal{E}^{r,\sigma}$. This defines a half-infinite path parallel to the positive/negative $r$ principal direction and starting/ending at $v$.\\

    Let $p_{(r_1, \sigma_1, \tau_1), (r_2, \sigma_2, \tau_2)}$ denote an infinite monotonic path such that $p_{(r_1, \sigma_1, \tau_1), (r_2, \sigma_2, \tau_2)}(t) \in \dg_{\tau_1}^{(r_1,\sigma_1)}$ for all $t>t_+$, and $p_{(r_1, \sigma_1, \tau_1), (r_2, \sigma_2, \tau_2)}(t) \in \dg_{\tau_2}^{(r_2,\sigma_2)}$ for all $t<t_-$ for integer constants $t_\pm$.  Let $\mathcal{P}_{(r_1, \sigma_1, \tau_1), (r_2, \sigma_2, \tau_2)}$ denote the set of infinite paths that are path equivalent to $p_{(r_1, \sigma_1, \tau_1), (r_2, \sigma_2, \tau_2)}$.

\begin{restatable}{theorem}{threesectorclass}
    All {\em inequivalent} ground states containing 3 infinite flux strings $\dg_i$ (for $i = 1,2,3$) and a number of charges are labelled by $(g \in \mathbb{Z}_2, \dg_i \in \mathcal{P}_{(r_i, \sigma_i, \tau_i), (r_i, \overline{\sigma}_i, \tau_i')})$ where $r_i$ are unique elements of $\{x,y,z\}$ and $g$ indicates if the ground state is charged or uncharged.

\end{restatable}


\begin{restatable}{theorem}{foursectorclass}
    There does not exist any ground state with a configuration of infinite flux strings ${{\{\dg_n\}}_{n=1}^N}$ for $N\ge 4$.
\end{restatable}

We wish to comment that a classification exists for sectors containing 2 infinite flux strings. Though the classification is more involved and needs to be split up into smaller cases. For this reason we do not state it in the summary. We will fully work out the classification in section \ref{sec:2str}.

%% file: Main/GS_purely_charged.tex
\section{Purely charged/uncharged ground states}
\label{sec:pure-charge}
\subsection{Finite string/surface operators}
We first study the case when we have an uncharged ground state.  This state $\omega$ is defined to be $\omega(A_v) = \omega(B_f) = 1$ for all $v,f$. It is not hard to see that this state has the lowest possible energy, since it has an eigenvalue $+1$ for all $A_v, B_f$. The GNS triple for this state is denoted by $(\pi_\omega, \hilb_\omega, \Omega_\omega)$, with $\omega (O) = \langle \Omega_\omega, \pi_\omega (O) \Omega_\omega \rangle$. We take ${\Omega_\omega}$ as the vacuum vector, defined by the property $\pi_\omega(A_v) {\Omega_\omega} = \pi_\omega(B_f) {\Omega_\omega} = {\Omega_\omega}$ for all $v,f$.\\

\begin{lemma}
Equivalence of representations implies the corresponding GNS vectors lie in the same sector.
\end{lemma}
\begin{proof}
    Let $\omega_1, \omega_2$ be two states of $\cstar{}$ with the GNS triples $(\pi_1, \hilb_1, \Omega_1), (\pi_2, \hilb_2, \Omega_2)$ respectively. Equivalence of representations implies the existence of a unitary map $U:\hilb_1 \rightarrow \hilb_2$ such that $\pi_2 = U \pi_1 U^\dagger$. Then we have,
    \begin{align*}
        \omega_1 ({O}) &= \langle \Omega_1, \pi_1 ({O}) \Omega_1\rangle \qquad O \in \cstar{} \\
        &= \langle \Omega_1, U^\dagger \pi_2 ({O}) U \Omega_1 \rangle\\
        &= \langle  U\Omega_1, \pi_2 ({O}) U \Omega_1 \rangle\\
        &= \omega_2 ({O})
    \end{align*}
    With the new GNS triple $(\pi_2, \hilb_1, U \Omega_1)$, but such a triple is unique up to equivalence. We see that the states both live in the same sector. 
\end{proof}

A 2-charge state $\omega_\gamma = \langle \Omega_\omega, \pi_\omega \circ \alpha_\gamma(\cdot) \Omega_\omega \rangle$ can be built for a finite path $\gamma$ following section \ref{sec:constexcite}. The unitary $U = \pi_\omega(F_\gamma)$ gives us the equivalence relation $\pi_\omega= U (\pi_\omega \circ \alpha_\gamma) U^\dagger$. Thus the two states lie in the same sector $\hilb_\omega$. An analoguous analysis exists for the state with finite flux excitations. We recall $\omega_\dg$ as the state with the representation $\pi_\dg$ corresponding to a finite flux string along $\dg$ as defined in section \ref{sec:constexcite}.

\begin{proposition}
\label{pro:EnergyBdy}
    The energy of a state $\omega_\dg$ having a single finite flux string is proportional to $|\dg|$. The energy of a flux string inside region $\Lambda$ is proportional to the number of edges in $\dg$ that lie inside $\Lambda$.
\end{proposition}
\begin{proof}
    Consider $\omega_\dg$ having a finite flux string, with $\dS$ being the bounding surface of $\dg$. We have $B_f F_\dS = - F_\dS B_f$ if $f$ is dual to some $\duale \in \dg$, and commute otherwise. This implies for such $f$, $\omega_\dg(B_f) = -1$. Consider a finite region $\overline{\Lambda} \subset \dG$ such that $\dS$ is entirely contained in $\overline{\Lambda}$. Then the energy in $\overline{\Lambda}$ is given by $\omega_\dg(H_{\overline{\Lambda}}) = 2|\dg|$. If $\overline{\Lambda}$ does not entirely contain $\dS$, then the same calculation will give us $\omega_\dg(H_{\overline{\Lambda}}) =2|\dg|_{\overline{\Lambda}}  $ where $|\dg|_{\overline{\Lambda}}:=2|\dg \setminus \mathcal{E}(\overline{\Lambda}^c)|$ is the \# of edges of $\dg$ inside $\overline{\Lambda}$ and $\overline{\Lambda}^c \subset \dG$ is the complement of $\overline{\Lambda}$.
\end{proof}

\begin{proposition}
    Finite flux strings are always closed paths on the dual lattice.
\end{proposition}
\begin{proof}
Let $\dS$ be a surface with the smallest boundary. It contains only a single face, and $|\dS| = 1$. The boundary of a face is a closed path. Since all surfaces are products of faces, the boundary of products of faces is given by pathwise addition of boundaries of the individual faces. But the pathwise addition of closed paths is always closed. It follows that since finite flux strings are the boundary excitations of a surface operator $F_\dS$, they are closed paths on the dual lattice $\dG$.
\end{proof}

\subsection{Constructing a charged sector $\hilb_\epsilon$}
We can construct states in a charged sector by considering automorphisms $\alpha_v$ defined in section \ref{sec:recap}. We restate the definition here:
\begin{align}
    \alpha_v(O) = \lim_{n \rightarrow \infty} F_{\gamma_v(n)}OF_{\gamma_v(n)} \qquad O \in \cstar{}
\end{align}
Where $\gamma_v(n)$ is a path that starts at $v$, stretches down in the $-\hat{z}$ direction and ends at $v - n \hat{z}$.

\begin{remark}
    The particular limit of $n \rightarrow \infty$ was an arbitrary choice. In the physics literature this is referred to as a gauge choice. In principle, any sufficiently nice path \footnote{The representation $\pi_v$ should lie inside a 3d conelike region. For a thorough 2d treatment, refer to \cite{Naaijkens2010-aq}.} stretching to infinity in an arbitrary direction would have worked.
\end{remark}

\begin{remark}
    We prove the convergence of $\alpha_v$ in lemma \ref{lem:aut_conv} in appendix \ref{app:lattice}.
\end{remark}

\begin{theorem}[\cite{Naaijkens2013-ji}, proposition 3.2.8]
\label{thm:equiv}
    Let $\cstar{}$ be a quasilocal algebra of some spin system, and suppose $\omega_1, \omega_2$ are pure states on $\cstar{}$. Then the following criteria are equivalent:\\

    \begin{itemize}
        \item The corresponding GNS representations $\pi_1, \pi_2$ are equivalent.
        \item For each $\epsilon>0$, there is a $\Lambda_\epsilon \in \Lambda_f$ such that $$|\omega_1(O) - \omega_2(O)| < \epsilon ||O||$$ for all $O \in \cstar{\Lambda}$ and $\Lambda$ a finite region in $\Lambda_\epsilon^c$.
    \end{itemize}
\end{theorem}

\begin{lemma}
    Representations $\pi_v:= \pi_\omega \circ \alpha_v, \pi_\omega$ are inequivalent
\end{lemma}
\begin{proof}
    Consider $\omega_v := \omega \circ \alpha_v$. Let its GNS representation be $\pi_v$. We will use theorem \ref{thm:equiv} to prove this lemma. Consider a spherical region $\Lambda_\epsilon$ centered at $v$. We can then consider $A = F_{\dS} \in \cstar{\Lambda}$ as a flux operator on a closed surface $\dS$ going around $\Lambda_\epsilon$, as shown in figure \ref{fig:spherical_regions}. Here $\Lambda$ is a finite region in $\Lambda_\epsilon^c$. We then have:
    \begin{align*}
        |\omega \circ \alpha_v (F_{\dS}) - \omega(F_\dS)| &= |\lim_{n \rightarrow \infty} \omega(F_{\gamma_v(n)}F_{\dS} F_{\gamma_v(n)}) - \omega(F_\dS)|\\
        &= |\lim_{n \rightarrow \infty} -\omega(F_{\dS}F_{\gamma_v(n)} F_{\gamma_v(n)}) - \omega(F_\dS)|\\
        &= | -\omega(F_{\dS}) - \omega(F_\dS)|\\
        &= 2 ||F_{\dS}|| 
    \end{align*}
    Which is independent of $\epsilon$. From theorem \ref{thm:equiv} $[\pi_v] \neq [\pi_\omega]$. This concludes our proof. 
\end{proof}

\begin{figure}[t!]
    \centering

        \centering
    \begin{subfigure}[t]{0.45\textwidth}
        \centering
\includegraphics[ width=0.4\textwidth]{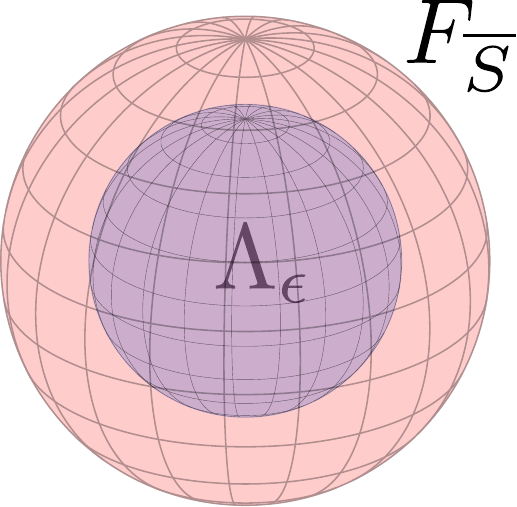}
    \caption{A blue spherical region $\Lambda_\epsilon$ and a flux operator $F_\dS$ supported on $\Lambda_\epsilon^c$ colored in red.}
    \label{fig:spherical_regions}

    \end{subfigure}%
    ~ 
    \begin{subfigure}[t]{0.45\textwidth}
        \centering
\includegraphics[ width=0.4\textwidth]{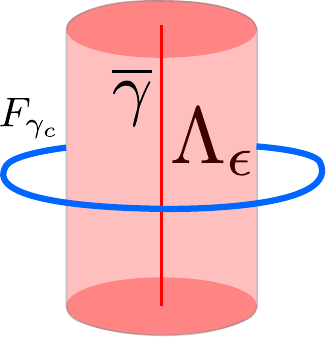}
    \caption{A red cylindrical region $\Lambda_\epsilon$ and the string operator $F_{\gamma_c}$ supported on $\Lambda_\epsilon^c$ colored in blue.}
    \label{fig:cylindrical_regions}
    \end{subfigure}

    \caption{}

\end{figure}

These automorphisms are involutary ($\alpha_v^2 = \mathds{1}$) and translation covariant ($T_x \circ \alpha_v = \alpha_{v+x} \circ T_x$) for $x \in \mathcal{V}$. Different automorphisms are related to each other via a unitary transformation: $$U \alpha_v U^\dagger = \alpha_{v+x}$$ where $U = F_{\gamma(v,v+x)}$ and $\gamma(v,v+x)$ is a path from $v$ to $v+x$.\\

There is a family of purely charged ground states in $\hilb_\epsilon$, given by $$\omega_v (O) := \langle \Omega_\omega, \pi_v (O) \Omega_\omega\rangle$$ the ground states are distinguished by $A_{v'}$: $\omega_v (A_{v'}) = 1 -2\delta_{v,v'}$. 

\begin{remark}
\label{rem:aut_commute}
    $\alpha_v \circ \alpha_{v'} = \alpha_{v'} \circ \alpha_{v}$ since charge operators commute. Similarly, $\alpha_{\dg} \circ \alpha_{\dg'} = \alpha_{\dg'} \circ \alpha_{\dg}$ since the flux strings commute. 
\end{remark}

\begin{remark}
    $\alpha_v \circ \alpha_{\dg} = \alpha_{\dg} \circ \alpha_v$ for a finite flux string $\dg$. This is a special case of lemma \ref{lem:aut_commute}, which we will prove in the next subsection.
\end{remark}

We now recall and prove the theorem \ref{thm:charged-sectors} for purely charged states.

\purelycharged*
\begin{proof}
    Let a general state with $n$ charges and $m$ finite flux strings $\dg_i$ be given by $$\omega_{v_1,\cdots, v_n; \dg_1,\cdots, \dg_m}(O):= \langle \Omega_\omega, \pi_{v_1,\cdots, v_n; \dg_1,\cdots, \dg_m} O \Omega_\omega \rangle \qquad (O) \in \cstar{}$$ where $\pi_{v_1,\cdots, v_n; \dg_1,\cdots, \dg_m}:= \pi_\omega \circ \alpha_{v_1} \circ \cdots \circ \alpha_{v_n} \circ \alpha_{\dg_1} \circ \cdots \circ \alpha_{\dg_m}$. \\

    We have $[\pi_{v_1,\cdots, v_n; \dg_1,\cdots, \dg_m}] = [\pi_{v_1,\cdots, v_n}]$ as the flux strings are finite. If $n$ is even, then $[\pi_{v_1,\cdots, v_n}] = [\pi_\omega]$ with the unitary given by $U = \pi_\omega(\prod_{i=2}^n F_{\gamma(v_i,v_1)})$. If $n$ is odd, $[\pi_{v_1,\cdots, v_n}] = [\pi_v]$ with the unitary given by $U = \pi_\omega(\prod_{i=1}^n F_{\gamma(v_i,v)})$. Since $[\pi_\omega] \neq [\pi_v]$, there are two possible ground state sectors: $\mathcal{H}_e, \mathcal{H}_\omega$.
\end{proof}

%% file: Main/GS_on_1str.tex
\section{1 string configurations}
\label{sec:1str}
Through the remainder of this paper, we will work with the dual lattice $\dG$ as it is more convenient for constructing the infinite flux string states. Since a dual-path in $\Gamma$ is just a path in $\dG$, we will be referring to $\dg \in \dG$ as paths for the sake of brevity. However, to not confuse the reader, we still choose to denote objects in the dual lattice by $(\overline{\cdot})$.\\

In this section, we will aim to first build an infinite flux string state $\omega_\dg$ using infinite surface automorphisms. We will then focus on determining the necessary and sufficient conditions for $\omega_\dg$ to be a ground state. Finally we will try to classify the ground state sectors with a single path $\dg$.
\subsection{Building an infinite flux string state}
Building a new sector is a little more involved for flux strings. Let's first understand how to build a state in the new sector using an example. Consider an infinite path $\dg$ as defined in definition \ref{def:infpath} going in the $+\hat{z}$ direction, as depicted in figure \ref{fig:parallel_to_z}. We will refer to figure \ref{fig:finite_surface_for_gamma} for the proceeding construction of the infinite flux string $\dg$.\\

We start by defining a series of finite surface operators $\dS_{\dg_n}$, that have a finite section $\dg_n$ of $\dg$ as a part of their boundary. We also have $\dS_{\dg_n} \subset \dS_{\dg_{n+1}}$. Upon taking the limit $n \rightarrow \infty$ we obtain $\dg$ as the only finite boundary of $\lim_{n\rightarrow \infty} \dS_{\dg_n}$.\\

\begin{figure}[t!]
    \centering

        \centering
    \begin{subfigure}[t]{0.5\textwidth}
        \centering
\includegraphics[ width=0.5\textwidth]{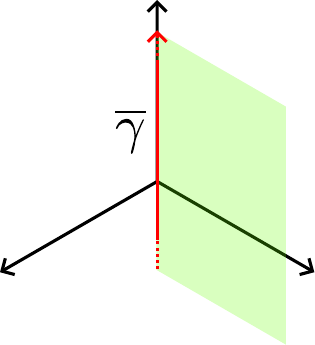}
\caption{The configuration we wish to achieve}
    \label{fig:parallel_to_z}

    \end{subfigure}%
    ~ 
    \begin{subfigure}[t]{0.5\textwidth}
        \centering
\includegraphics[ width=0.5\textwidth]{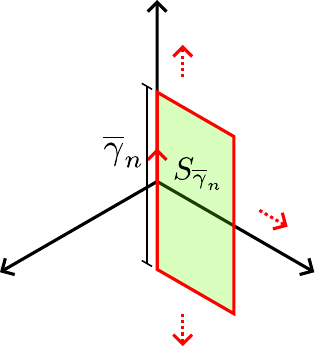}
\caption{The surface whose boundary will give us finite sections of $\dg$}

    \label{fig:finite_surface_for_gamma}
    \end{subfigure}

\end{figure}


Consider now the automorphism $$\alpha_{\dg}(O) := \lim_{n\rightarrow \infty} F_{\dS_{\dg_n}} O F_{\dS_{\dg_n}} \qquad O \in \cstar{}$$ we can obtain a new state $\omega_\dg := \langle \Omega_\omega, \pi_\dg (\cdot) \Omega_\omega \rangle$ where $\pi_\dg := \pi_\omega \circ \alpha_\dg$.

\begin{remark}
    We prove the convergence of $\alpha_\dg$ in lemma \ref{lem:aut_conv_2} in appendix \ref{app:lattice}.
\end{remark}

\begin{lemma}
    Representations $\pi_\dg, \pi_\omega$ are inequivalent.
    \label{lem:path_inequiv}
\end{lemma}
\begin{proof}
    Consider $\omega_\dg := \omega \circ \alpha_\dg$. Let its GNS representation be $\pi_\dg$. Consider a cylindrical region $\Lambda_\epsilon$ centered around $\dg$. We can then consider $A = F_{\gamma_c} \in \cstar{\Lambda}$ as a charge operator on a closed string $\gamma_c$ going around $\Lambda_\epsilon$, as shown in figure \ref{fig:cylindrical_regions}. Here $\Lambda$ is a finite region in $\Lambda_\epsilon^c$. We then have:
    \begin{align*}
        |\omega \circ \alpha_\dg (F_{\gamma_c}) - \omega(F_{\gamma_c})| &= |\lim_{n \rightarrow \infty} \omega(F_{\dS_{\dg_n}} F_{\gamma_c} F_{\dS_{\dg_n}}) - \omega(F_{\gamma_c})|\\
        &= |\lim_{n \rightarrow \infty} -\omega(F_{\gamma_c}F_{\dS_{\dg_n}}F_{\dS_{\dg_n}}) - \omega(F_{\gamma_c})|\\
        &= | -\omega(F_{\gamma_c}) - \omega(F_{\gamma_c})|\\
        &= 2 ||{F_{\gamma_c}}|| 
    \end{align*}
    Which is independent of $\epsilon$. From theorem \ref{thm:equiv} $[\pi_\dg] \neq [\pi_\omega]$. This concludes our proof. 
\end{proof}

$\omega_\dg$ thus lies in a new sector, which we denote by $\hilb_\dg$.\\

The bounding surface needs to be defined for the specific infinite flux line $\dg$. This is known as choosing the gauge. Picking a surface is a matter of choice, as the operators in $\cstar{}$ cannot physically detect the bounding surface. This can always be done as long as one can find a 3d wedge that $\dg$ lies entirely within. There is no universally "nice" gauge choice for a generic infinite path $\dg$ contrasting the case for charged excitations.

\subsection{Necessary and sufficient conditions for a ground state}
\label{sec:CondForGS1}

\begin{lemma}
\label{lem:PathEqvSameSector}
    If two paths $\dg_1, \dg_2$ are path equivalent, then $[\pi_{\dg_1}] = [\pi_{\dg_2}]$.
\end{lemma}
\begin{proof}
If $\dg_1, \dg_2$ are path equivalent, then we can construct a finite surface $\dS$ whose boundary is $\partial \dS = (\dg_1 \cup \dg_2) \setminus (\dg_1 \cap \dg_2)$. We then have the required unitary $U = \pi_\omega (F_\dS)$ such that $\pi_{\dg_1} = U \pi_{\dg_2} U^\dagger$.    
\end{proof}

\begin{lemma}
\label{lem:energy-diff}
    If two paths $\dg_1, \dg_2$ are path equivalent, then one can choose a finite region $\overline{\Lambda}$ which satisfies $\dg_1 (t) \in \dg_2$ for all $\dg_1 (t) \in \edge(\overline{\Lambda}^c)$. The energy difference $\Delta E$ between $\omega_{\dg_1},\omega_{\dg_2}$ is given by $\Delta E = \omega_{\dg_1}(H_{\overline{\Lambda}}) - \omega_{\dg_2} (H_{\overline{\Lambda}})$. $\Delta E$ is finite and independent of the choice of region $\overline{\Lambda}$.
\end{lemma}
\begin{proof}
    Let $\dg_1, \dg_2$ be path equivalent. Then there exist $r,s$ with $r < s$ such that for all $t<r, t>s$, $\dg_1 (t) \in \dg_2$. We can similarly define $r',s'$ with $r' < s'$ such that $\dg_2(t') \in \dg_1$ for all $t'<r', t'> s'$. Choose a finite region $\overline{\Lambda}$ such that $\dg_1(t), \dg_2(t') \in \mathcal{E}(\overline{\Lambda})$ for all $r<t<s$ and $ r'<t'<s'$ and $\dg_1(t), \dg_2(t') \in \mathcal{E}(\overline{\Lambda}^c)$ otherwise. This is the required region where $\dg_1, \dg_2$ share all edges outside of $\overline{\Lambda}$. It is minimal in the sense that only the edges in $\dg_1, \dg_2$ that are different are in $\mathcal{E}(\dL)$ and the shared edges are in $\mathcal{E}(\dL^c)$.\\

     Choose a region $\dL' \supset \dL$. This region will also satisfy $\dg_1(t) \in \dg_2$ for all $\dg_1(t) \in \mathcal{E}(\dL')$. We can always split this into two disjoint subsets: $\dL' = \dL \cup \dL''$. The energy difference inside $\dL'$ is given by 

     \begin{align*}
         \Delta E_{\dL'} &= \omega_{\dg_1}(H_{\dL'}) - \omega_{\dg_2}(H_{\dL'})\\
         &= |\dg_1|_{\dL'} - |\dg_2|_{\dL'}\\
         &= |\dg_1|_{\dL} - |\dg_2|_{\dL}
     \end{align*}
     Where $|\dg|_{\dL}$ was defined in proposition \ref{pro:EnergyBdy}. The last equality follows from the fact that $\dg_1, \dg_2$ share all their edges inside $\dL''$. So the energy difference between $\omega_{\dg_1}, \omega_{\dg_2}$ is independent of the region $\dL'$, and is finite.
\end{proof}


We can now begin the program of proving the theorem \ref{thm:mono-GS}. Recall definition \ref{def:monopath} of a monotonic infinite path. Figure \ref{fig:mono_takes_least} shows some examples of finite sections of monotonic and non-monotonic paths.

\begin{figure}
    \centering
    \includegraphics[width=0.3 \textwidth]{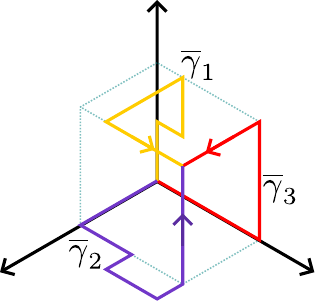}
    \caption{An example of different monotonic and non-monotonic paths between two endpoints of $\dL$. $\dg_3$ is non-monotonic while $\dg_1, \dg_2$ are monotonic.}
    \label{fig:mono_takes_least}
\end{figure}

\begin{lemma}
\label{lem:MonoDefGS}
    Consider a monotonic infinite path $\dg$. Choose a finite cuboidal region $\overline{\Lambda}$ such that $\partial_0\dg(n)$ and $\partial_1 \dg(m)$ lie on the longest diagonal of $\overline{\Lambda}$ for arbitrary integer constants $m,n$. Inside $\overline{\Lambda}$, all monotonic deformations $\dg'$ of $\dg$ have the same energy. All such $\omega_{\dg'}$ have the lowest energy inside $\dL$.
\end{lemma}
\begin{proof}
WLOG assume $\overline{\Lambda}$ is a cuboidal region $[0,a] \times [0,b] \times [0,c]$. Consider a finite monotonic path such that $\partial_0\dg(0) = (0,0,0)$ and $\partial_1 \dg(N) = (a,b,c)$. Let $N_r = \#\{\dg(t) | \dg(t) \in \edge(r,\pm)\}$ for $r\in \{x,y,z\}$. To go from $(0,0,0)$ to $(a,b,c)$ one needs at least $a+b+c$ edges.\\

Since $\dg$ is monotonic in $x,y,z$ directions, it takes the least amount of edges to get to $(a,b,c)$. So $N_x = a, N_y = b,N_z = c$ . If there are any edges that don't go towards $(a,b,c)$, as is the case in a non-monotonic path, that path would be longer. All monotonic deformations $\dg'$ of $\dg$ inside $\dL$ will also take the least number of edges on account of being monotonic. $\omega_\dg(H_{\overline{\Lambda}}) = |\dg|_{\overline{\Lambda}}$ from proposition \ref{pro:EnergyBdy}. As all $\dg'$ already have the least $\#$ of edges, we have that $\omega_{\dg'}$ all have the lowest energy inside $\dL$ if the endpoints are held fixed.
\end{proof}

\begin{theorem}
\label{thm:LowEState}
    For a given path $\dg$, $\omega_\dg$ is not a ground state if $\dg$ is non-monotonic. 
\end{theorem}


\begin{figure}
    \centering
    \includegraphics[width = \textwidth]{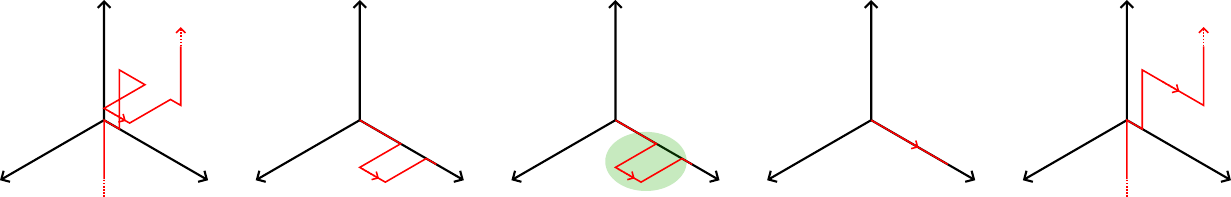}
    \caption{Straightening an example path that is non-monotonic in 2 directions (x,z in this figure). First project to $x-y$ plane. Then straighten in that plane. Then lift back up to entire lattice.}
    \label{fig:straightening}
\end{figure}

\begin{proof}
    Suppose $\dg$ is a non-monotonic path inside a finite region $\dL$ (and possibly outside $\dL$ as well). Then one can instead consider a path $\dg'$ which is the same as $\dg$ outside $\dL$, but monotonic inside $\dL$. This can always be done because there always exists at least one monotonic path between any two points inside $\dG$. Due to proposition \ref{pro:EnergyBdy}, $\dg'$ has a lower energy than $\dg$, so $\omega_\dg$ is not a ground state.

\end{proof}
Let a non-monotonic path $\dg$ be given. We explicitly sketch the construction of a shorter path $\dg'$. We will call this construction the "straightening procedure", or "straighten" in short.\\

    Let $\mu, \nu \in \{x,y,z\}$ be two directions. We can divide the straightening procedure into 3 cases of increasing complexity:\\

    (I)  First consider the case where $\dg$ is parallel to one of the principal planes, call it the $\nu$- plane, such that $\nu$ is the normal to the plane. Assume that it is non-monotonic only along the $\mu$ direction. Then there must exist two edges $\dg(m), \dg(n)$ such that $(\partial_0 \dg(m))_\mu = (\partial_1 \dg(n))_\mu$. Then we can construct a new path $\dg'$ such that $\dg'(t) \in \dg$ for all $t<m, t>n$, and it is monotonic between points $\partial_0 \dg(m), \partial_1 \dg(n)$. $\dg'$ is path equivalent to $\dg$ and is monotonic in this region. Thus $\omega_{\dg'}$ must have a lower energy than $\omega_\dg$ so it couldn't have been the ground state. \\

    (II) Now we consider the case where $\dg$ is still non-monotonic along the $\mu$ direction, but is now no longer restricted to lie parallel to the $\nu$ plane. In this case we can define a projection $P_\nu$ that projects $\dg$ on the $\nu$ plane. This can be done by simply throwing out the edges in the $\nu$ direction. The path $\dg_\nu$ thus obtained may now be finite. We can apply the same procedure as (I) to obtain a monotonic path $\dg_\nu'$. We can then lift this path back into the $\nu$ direction by reintroducing the edges we threw out to obtain $\dg'$, which is path equivalent to $\dg$ but is monotonic inside this region.\\

    (III) If $\dg$ is non-monotonic along more than one direction, then one can always choose a smaller region where it is monotonic in 2 of the 3 principal directions.

\begin{corollary}
\label{cor:InfLowE}
    We can use the straightening procedure to construct another state $\omega_{\dg'}$ with lower energy inside $\dL$ than a state $\omega_\dg$ with a  non-monotonic path $\dg$. If we end up with a monotonic path $\dg$ after straightening a finite number of times in different regions, $\dg'$ is path-equivalent to $\dg$, and by lemma \ref{lem:PathEqvSameSector} $\omega_{\dg_1}, \omega_{\dg_2}$ are equivalent.
\end{corollary}
\begin{proof}
    Let $\dg'$ be a result of straightening $\dg$ N times. Let $\dL_n$ be the finite region that encompasses the section of $\dg$ that was straightened in the $n$th time. Then we can consider a region $\dL =\cup_{n=1}^N \dL_n$. Since $\dg'$ was obtained from $\dg$ through modifications inside $\dL$, we have $\dg'(r) \in \dg$ for all $\dg'(r) \notin \edge(\dL)$. Hence $\dg'$ is path equivalent to $\dg$.
\end{proof}

Let us now recall and prove theorem \ref{thm:mono-GS}:
\monotonicGS*
\begin{proof}
    Consider a path $\dg$. If it is non-monotonic, theorem \ref{thm:LowEState} shows $\omega_\dg$ is not a ground state. If it is monotonic, then from lemma \ref{lem:MonoDefGS} it has the lowest energy inside any given cuboidal region $\dL$ such that $\partial_0\dg(n)$, $\partial_1\dg(m)$ lie on the endpoints of a longest diagonal of $\dL$. Since we can freely choose said $\dL$, $\omega_\dg$ must be a ground state.\\

    Conversely, If $\omega_\dg$ is a ground state, then it must have the lowest energy and thus least number of edges. Through lemma \ref{lem:MonoDefGS} we know that a path that takes the fewest edges to get from start to end point inside $\dL$ must be monotonic. So $\omega_\dg$ being a ground state implies $\dg$ is monotonic. 
\end{proof}

\begin{figure}[t!]
    \centering

        \centering
    \begin{subfigure}[t]{0.5\textwidth}
        \centering
\includegraphics[ width=0.5\textwidth]{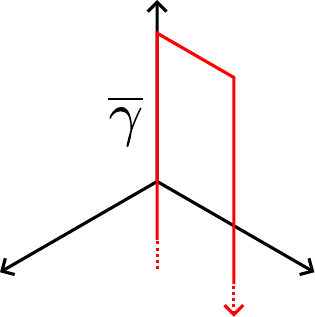}

    \end{subfigure}%
    ~ 
    \begin{subfigure}[t]{0.5\textwidth}
        \centering
\includegraphics[ width=0.5\textwidth]{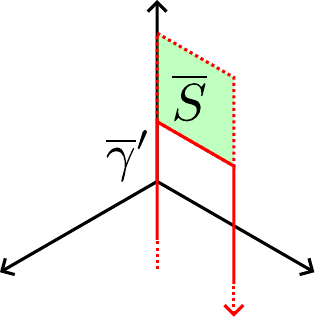}

    \end{subfigure}
    \caption{We can keep straightening this configuration}
    \label{fig:inverse_U}
\end{figure}


As an example of the straightening procedure, let us explicitly start from a non-monotonic state construct a new state with lower energy. Consider a state $\omega_\dg$ with a U shaped path $\dg$ as shown in figure \ref{fig:inverse_U}. Straighten $\dg$ inside region $\dL$ as shown in the figure to build a new state $\omega_{\dg(n)}$ with a path $\dg(n)$. We can indefinitely straighten it by choosing a different region $\dL(n)$. The limit in $n$ of this procedure will give us $\omega$, the state with no flux strings. However, the representations $\pi_\dg$ and $\pi_\omega$ are inequivalent from lemma \ref{lem:path_inequiv}. So $\omega_\dg$ does in-fact belong to a different sector. However this is not a ground state sector as the energy of $\omega_\dg$ can be lowered indefinitely.


\subsection{Infinity directions and a classification}
We talk about a classification of ground states. We start by considering Infinity directions (recall definition \ref{def:InfDir}) as our tool for understanding which states could be reduced to ground states. 

\begin{lemma}
\label{lem:redefinfdir}
     The state $\omega_\dg$ does not have a canonical labelling of positive or negative infinity directions. Positive and negative infinity directions are thus interchangeable for any path $\dg$.
\end{lemma}

\begin{proof}
    Consider a path $\dg$ whose edges are labelled by an integer parameter $t$. The transformation $t\mapsto -t$ reverses the orientation of $\dg$ and maps $D_+ (\dg)$ to $D_- (\dg)$ and $D_- (\dg)$ to $D_+ (\dg)$. However the orientation of $\dg$ is irrelevant since $\alpha_\dg ^\dagger = \alpha_\dg$. So there is no canonical labelling and $D_- (\dg)$ and $D_+ (\dg)$ are interchangeable.
\end{proof}

\begin{lemma}
\label{lem:region}
    Let $D(\dg)$ be the infinity directions for path $\dg$. One can always choose a finite region $\dL$ such that for all $\dg(t) \in \mathcal{E}(\dL^c)$ we have $\dg(t) \in \cup_{(r,\sigma) \in D(\dg)} \mathcal{E}^{r,\sigma}$
\end{lemma}
\begin{proof}
    Since for all $(r,\sigma) \notin D(\dg)$ we have $\#\{\dg(t)| \dg(t) \in \mathcal{E}^{r,\sigma}\} < \infty$, we can choose a finite region $\dL$ such that $\dg(t) \in \mathcal{E}(\dL)$ for all $\dg(t) \in \cup_{r,\sigma \notin D(\dg)}\mathcal{E}^{r,\sigma}$. It follows that for all $\dg(t) \in \mathcal{E}(\dL^c)$ we have $\dg(t) \in \cup_{(r,\sigma) \in D(\dg)} \mathcal{E}^{r,\sigma}$.
\end{proof}

\begin{lemma}
\label{lem:InfDirGS1}
    $\omega_\dg$ is in a ground state sector iff $D_+(\dg) \cap D_-(\dg) = \emptyset$
\end{lemma}

\begin{proof}
    Consider a state $\omega_\dg$ such that $\dg$ has a direction $(r,\sigma) \in D_+(\dg) \cap D_-(\dg)$. \\

    There will then exist an integer $c = (\partial_0 \dg(n))_r = (\partial_1 \dg(m))_r$. This means $\dg$ is non-monotonic in at least the $r$ direction. We can then straighten it to a path $\dg'$ with a lower energy. But since there are infinite such edges, one can always consider another integer $c' = (\partial_0 \dg(n'))_r = (\partial_1 \dg(m'))_r$ such that $\sigma c' > \sigma c$. Thus we can construct paths of lower energy indefinitely, implying we cannot reach a path $\dg'$ for which $\omega_{\dg'}$ is a ground state.\\

    We can easily see the converse from theorem \ref{thm:mono-GS}. If $\omega_\dg$ is a ground state, it must be monotonic. So there exists $\sigma_r \in \{\pm\}$ such that $\dg \subseteq \cup_{r=1}^3 \edge^{r,\sigma_r}$. Then $D_+(\dg) \subseteq \{(r,\sigma_r)\}_{r\in\{x,y,z\}}$. Whereas $D_-(\dg)\subseteq \{(r,\sigma_r^c)\}_{r\in\{x,y,z\}}$ where $\sigma^c_r\in \{\pm\}$ is the complement of $\sigma_r$. It follows that $D_+(\dg) \cap D_-(\dg) =\emptyset$.
\end{proof}
\begin{definition}[\textit{Pathologically non-monotonic}]
    A path $\dg$ is pathologically non-monotonic if  there exists an infinite sequence $\{\dg_n\}_{n = 1}^{\infty}$ such that $\dg_1 = \dg$, $\dg_n$ is path equivalent to $\dg_m$ for all $m,n \in \mathbb{Z}_+$, and we can straighten $\dg_n$ to $\dg_{n+1}$ for all $n \ge 1$ resulting in a state $\omega_{\dg_{n+1}}$ with lower energy than $\omega_{\dg_n}$.
\end{definition}

\begin{corollary}
\label{cor:NotGSSIfInfDir}
    If $\dg$ is pathologically non-monotonic, then $D_+(\dg) \cap D_-(\dg) \neq \emptyset$. Consequently, $\omega_\dg$ is not in a ground state sector.
\end{corollary}
\begin{proof}
    If $\dg$ is pathologically monotonic, then we have an infinite sequence $\{\dg_n\}$  such that $\dg_n$ can be straightened to $\dg_{n+1}$ such that $\omega_{\dg_{n+1}}$ has a lower energy than $\omega_{\dg_n}$. Call the direction along which we straighten $\dg_n$ to $\dg_{n+1}$ as $(r_n, \sigma_n)$.\\

Let $N_{(r,\sigma)} = \#\{(r_n,\sigma_n)|n \in \mathbb{Z}_+, r_n = r, \sigma_n = \sigma\}$, where $r \in \{x,y,z\}, \sigma \in {\pm}$. Since $\#\{(r_n, \sigma_n)\} = \infty$, there must exist at least one direction $(r,\sigma)$ such that $N_{(r,\sigma)} = \infty$. \\

Pick $\dg_m$ such that $(r_m,\sigma_m) = (r,\sigma)$. Follow the straightening procedure from $\dg$ to $\dg_m$. Since $\dg$ can be straightened to $\dg_m$, there must exist two edges $\dg_m(a), \dg_m(b) \in \dg$ such that  $(\partial_0\dg_m(a))_r = (\partial_1 \dg_m(b))_r$. Let $g,h \in \mathbb{Z}$ be such that $\dg(g) = \dg_m(a), \dg(h) = \dg_m(b)$. Now pick an edge $\dg(t)$ of $\dg$ such that $g<t<h$. \\

Since $N_{r,\sigma} = \infty$, we must have another $\dg(m')$ with $m'>m$ such that $\dg$ can be straightened to it. We now have $\dg_{m'}(a), \dg_{m'}(b) \in \dg$ such that  $(\partial_0\dg_{m'}(a))_r = (\partial_1 \dg_{m'}(b))_r$. Let $g',h' \in \mathbb{Z}$ be such that $\dg(g') = \dg_{m'}(a'), \dg(h') = \dg_{m'}(b')$. We then have $g' < t < h'$. As $\dg_{m}$ can also be straightened to $\dg_{m'}$, we have $g' < g, h' > h$. This implies the direction $(r,\sigma)$ belongs to both $D_+(\dg), D_-(\dg)$. Hence $D_+(\dg) \cap D_-(\dg) \neq \emptyset$.\\
    
    Using lemma \ref{lem:InfDirGS1}, it follows that $\omega_\dg$ is not in a ground state sector. This concludes our proof.
\end{proof}

\begin{lemma}
    $\alpha_v \circ \alpha_\dg = \alpha_\dg \circ \alpha_v$, that is the automorphisms $\alpha_v, \alpha_\dg$ commute.
    \label{lem:aut_commute}
\end{lemma}
\begin{proof}
    Let $O \in \cstar{}$. Then we have,
    
    \begin{align*}
        \alpha_v \circ \alpha_\dg (O) &= \lim_{n\rightarrow \infty} \lim_{m \rightarrow \infty} F_{\gamma_v(n)} F_{\dS_{\dg(m)}} O F_{\dS_{\dg(m)}} F_{\gamma_v(n)}\\
        &= (-1)^{2\{\text{linking number}\}}\lim_{n\rightarrow \infty} \lim_{m \rightarrow \infty}  F_{\dS_{\dg(m)}} F_{\gamma_v(n)} O  F_{\gamma_v(n)} F_{\dS_{\dg(m)}}\\
        &= \alpha_\dg \circ \alpha_v (O)
    \end{align*}
    Where we have used the property that $F_{\gamma_v(n)} F_{\dS_{\dg(m)}} = (-1)^{\{\text{linking number}\}}F_{\dS_{\dg(m)}} F_{\gamma_v(n)}$.
\end{proof}

Lemma \ref{lem:aut_commute} implies that one does not need to worry about the order of the automorphisms, and to build a state with charges $v_i$ and fluxes $\dg_i$, one may choose a convention to first apply $\alpha_{\dg_i}$, then $\alpha_{v_i}$.

We now use infinity directions to attempt to classify the different ground state sectors of the 3d Toric Code. Let recall and prove theorem \ref{thm:onesectorclass}:

\onesectorclass*

\begin{proof}
    Consider a state $\omega_{\{v_i\}, \dg}$ with a number of charges at $v_i$ and an infinite flux string denoted by path $\dg$. We can immediately see (with the same reasoning as theorem \ref{thm:charged-sectors}) that $[\pi_{\{v_i\}, \dg}] = [\pi_{\dg}]$ if $|\{v_i\}|$ is even, and $[\pi_{\{v_i\}, \dg}] = [\pi_{v,\dg}]$ if $|\{v_i\}|$ is odd. $g\in \mathbb{Z}_2$ is the parity of $|\{v_i\}|$, indicating if the state is charged. From theorem \ref{thm:mono-GS}, if $\dg$ is a ground state then necessarily $\dg \in \mathcal{M}'$ where $\mathcal{M}'$ is the set of all monotonic paths. From \ref{lem:PathEqvSameSector}, any two states are equivalent if their paths are path equivalent. So to find all inequivalent ground states, we need $\dg \in \mathcal{M'}/\sim =: \mathcal{M}$ where $\sim$ indicates the equivalence class of path equivalent paths. Hence the inequivalent ground states are completely labelled by ($g \in \mathbb{Z}_2$, $\dg \in \mathcal{M}$). 

\end{proof}





%% file: Main/GS_on_2str.tex
\section{2 string configurations}
\label{sec:2str}

Using the formalism developed in section \ref{sec:1str} we can readily tackle the classification of 2 infinite flux string states after first defining the concepts of surgery and truncation. From theorem \ref{thm:mono-GS} we already know that for $\omega_{\dg_1, \dg_2}$ to be a ground state, we have a prerequisite condition that $\dg_1, \dg_2$ must individually be monotonic. However there are additional conditions that we will explore now.\\

\subsection{Performing surgery}
\label{subsec:surgery}
Consider as an example state $\omega_{\dg_1,\ \dg_2}$ with $\dg_1, \dg_2$ as two infinite monotonic paths as shown in figure \ref{fig:parallel}. Let us choose a rectangular finite surface $\dS$ as shown in the figure \ref{fig:double_U} such that we get $ \omega_{\dg,\dg_1, \dg_2} = \omega_{\dg_1', \dg_2'}$ where $\dg_1', \dg_2'$ are two disconnected U shaped paths and $\dg = \partial \dS$. Since ${\dg_1', \dg_2'}$ are both pathologically non-monotonic, using corollary \ref{cor:NotGSSIfInfDir} $\hilb_{\dg_1, \dg_2}$ is not a ground state sector.\\

\begin{figure}[t!]
    \centering

        \centering
    \begin{subfigure}[t]{0.5\textwidth}
        \centering
\includegraphics[ width=0.5\textwidth]{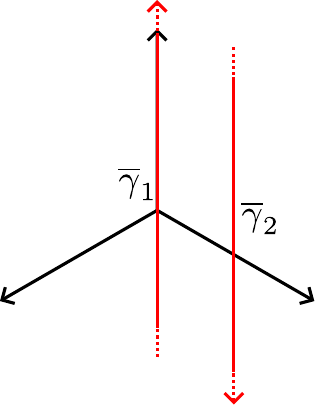}
\caption{}
\label{fig:parallel}
    \end{subfigure}%
    ~ 
    \begin{subfigure}[t]{0.5\textwidth}
        \centering
\includegraphics[ width=0.5\textwidth]{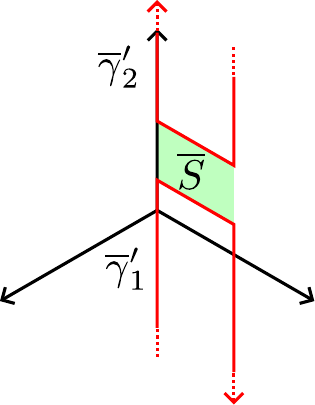}
\caption{}
\label{fig:double_U}

    \end{subfigure}
    \caption{Two parallel paths $\dg_1, \dg_2$ on which we perform surgery using $\dS$ which converts into a double U shaped configuration $\dg_1', \dg_2'$}
    \label{fig:parallelpaths}
\end{figure}

The above example can be called a "surgery" as we are in effect cutting the strings and reattaching them differently. In general we may perform surgery on any N infinite path configuration. It amounts to the following algorithm:

\begin{itemize}
    \item Let $\{\dg_n\}_{n=1}^N$ be a set of non-overlapping paths in the configuration.
    \item Consider a finite simply connected convex surface $\dS$ such that $\partial \dS \cap \dg_n \neq \emptyset$ for at least a single $n$.
    \item Performing surgery geometrically amounts to considering the set $(\partial \dS \bigcup_{n=1}^N \dg_n)\setminus \bigcup_{n=1}^N(\partial \dS \cap \dg_n)$
    \item     The set $(\partial \dS \bigcup_{n=1}^N \dg_n)\setminus \bigcup_{n=1}^N(\partial \dS \cap \dg_n)$ is a set of new paths $\{\dg_n'\}_{n=1}^{N}$. We prove this fact in theorem \ref{thm:set_of_paths}.
\end{itemize}

\begin{remark}
$\dS$ is finite, so surgery is performed using $\alpha_\dg \in \cstar{}$ and does not change the sector of configuration.
\end{remark}

\begin{remark}
    The set $(\partial \dS \bigcup_{n=1}^N \dg_n)\setminus \bigcup_{n=1}^N(\partial \dS \cap \dg_n)$ is specifically considered to reflect the $\mathbb{Z}_2$ grading of the flux operators $F_\dS$. This property geometrically means that two flux strings going over the same edge cancel out on that edge.
\end{remark}

\begin{theorem}
\label{thm:2monoGS}
    A state $\omega_{\dg_1, \dg_2}$ for $\dg_1, \dg_2$ as monotonic paths is in a ground state sector iff $D(\dg_1) \cap D(\dg_2) = \emptyset$
\end{theorem}
\begin{proof}
We first choose a convenient prescription for $D_\pm(\dg_1), D_\pm(\dg_2)$ taking advantage of lemma \ref{lem:redefinfdir} such that if there exists $(r,\sigma) \in D(\dg_1) \cap D(\dg_2)$, then $(r,\sigma) \in D_+(\dg_1), D_+(\dg_2)$. Since $\dg_1, \dg_2$ are monotonic, we have from lemma \ref{lem:InfDirGS1},
\begin{align}
\label{eqn:InfDirNoIntersect}
    &D_+(\dg_1) \cap D_-(\dg_1) = \emptyset & D_+(\dg_2) \cap D_-(\dg_2) = \emptyset
\end{align}
Let us perform surgery on $\dg_1, \dg_2$. Consider a finite region $\dL$ in accordance with lemma \ref{lem:region}. Choose a finite $\dS$ within $\dL$ such that $\partial \dS \cap \dg_i \neq \emptyset$ with $i = 1,2$. If the paths are non-overlapping then we will have $\dS \setminus (\dg_1 \cup \dg_2) \neq \emptyset$. If the paths are overlapping then we can always consider instead a finite deformation of these paths which makes them non-overlapping. Performing surgery will give us paths $\dg'_1, \dg'_2$ such that (upto relabelling of $D_\pm(\dg_i')$ using lemma \ref{lem:redefinfdir}) $D_+(\dg'_1) = D_+(\dg_1), D_- (\dg_1') = D_+(\dg_2)$ and $D_+(\dg'_2) = D_-(\dg_1), D_- (\dg_2') = D_-(\dg_2)$. If as previously assumed $(r,\sigma) \in D_+(\dg_1), D_+(\dg_2)$, then $D_+(\dg_1') \cap D_-(\dg_1') \neq \emptyset$. From lemma \ref{lem:InfDirGS1}, $\dg'_1$ is pathologically non-monotonic. So $\omega_{\dg_1, \dg_2}$ is not in a ground state sector. \\

To prove the converse, we have $D(\dg_1) \cap D(\dg_2) = \emptyset$. Since $\dg_1, \dg_2$ are monotonic, we also have eqn \ref{eqn:InfDirNoIntersect}. So even after performing surgery, we still have $D(\dg'_1) \cap D(\dg'_2) = \emptyset$. So we can only lower the energy of $\omega_{\dg'_1, \dg'_2}$ at most in a finite region $\dL$. It follows that $\omega_{\dg_1,\dg_2}$ is in a  ground state sector.
\end{proof}

Indeed, if we have a ground state $\omega_{\dg_1, \dg_2}$ where $\dg_1, \dg_2$ are two infinite monotonic paths, then $\omega_{\dg_1, \dg_2}$ lies in a ground state sector iff $D(\dg_1) \cap D(\dg_2) = \emptyset$. However, we have multiple cases corresponding to the different infinity directions of $\dg_1, \dg_2$, and we can say a little more about these cases if we consider them individually. \\

To classify the 2 infinity flux string case, we will have to first understand the solutions of $D(\dg_1) \cap D(\dg_2) = \emptyset$. We call the equation $D(\dg_1) \cap D(\dg_2) = \emptyset$ together with $D_+(\dg_i) \cap D_-(\dg_i) = \emptyset$ the ground state condition, GSC in short.

\subsection{Distinct solutions of the ground state condition}
\label{sec:find_all_classes}

 Let's divide the solutions into cases with increasing $|D(\dg_1)|, |D(\dg_2)|$. We also have the conditions $D_+(\dg_i) \cap D_-(\dg_i) = \emptyset$, and further $|D_\pm(\dg_i)| >0$ with $i = 1,2$. Let $\sigma \in \{\pm\}$. Define $\overline{\sigma} = \{\pm\} \setminus \sigma$.\\

We supplement each solution with a figure. We interpret the figures in the following way: there are 6 axes, each corresponding to a direction $(r,\sigma)$. The infinity directions for $\dg_1$ are given in red, while the ones for $\dg_2$ are given in blue. A red shaded plane or region means all infinity directions touching it belong to $\dg_1$ and similarly for the blue shaded region they belong to $\dg_2$. A shaded kite means it has 2 infinity directions touching it, and a shaded cube means it has 3 infinity directions touching it.\\

Case I: $|D(\dg_i)| = 2$: This is the easiest case. We first assume $D(\dg_1) = \{(r_1,\sigma_1), (r_2,\sigma_2)\}$ and $D(\dg_2) = \{(r_3  ,\sigma_3), (r_4,\sigma_4)\}$. Solving for GSC gives us $r_4 = r_1, \sigma_4 = \overline{\sigma}_1$. Notice that after surgery, simplification and relabelling, we can always obtain $D(\dg_1) = \{(r_1,\sigma_1), (r_1,\overline{\sigma}_1)\}$. Figure \ref{fig:D2} shows the solutions for case I.\\

\begin{figure}
    \centering
    \includegraphics[width = 0.25\textwidth]{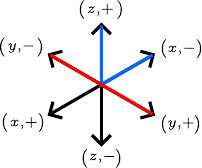}
    \caption{Solution for case I}
    \label{fig:D2}
\end{figure}

Case II: $|D(\dg_1)| = 2, |D(\dg_2)| = 3$: In this case, we have $D(\dg_1) = \{(r_1,\sigma_1), (r_2,\sigma_2)\}$ and $D(\dg_2) = \{(r_3  ,\sigma_3), (r_4,\sigma_4), (r_5, \sigma_5)\}$. Solving out for GSC will give us 3 separate solutions:
\begin{itemize}
    \item A: $r_1 \neq r_2$, $r_4 = r_1, \sigma_4 = \overline{\sigma}_1$, $r_5 = r_2, \sigma_5 = \overline{\sigma}_2$. Figure \ref{fig:D2=3,A} shows 3 solutions for case A. There are 3 additional solutions that are mirrored along the $x-z$ plane.

    \begin{figure}
    \centering
    \includegraphics[width = 0.8\textwidth]{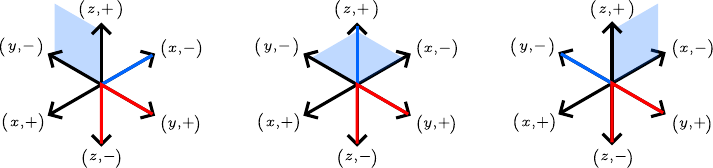}
    \caption{Solutions for case II.A (left), B (center), C (right).}
    \label{fig:D2=3,A}
\end{figure}


\item B: $r_1 \neq r_2$, $r_4 = r_3, \sigma_4 = \overline{\sigma}_3$, $r_5 = r_1, \sigma_5 = \overline{\sigma}_1$. Figure \ref{fig:D2=3,B} shows the solutions of case $B$.


    \item C: $r_1 = r_2, \sigma_2 = \overline{\sigma}_1$, $r_4 = r_3, \sigma_4 = \overline{\sigma}_3$. Figure \ref{fig:D2=3,C_after_surg} (left) shows the solutions for case C. After performing surgery it can be turned into \ref{fig:D2=3,C_after_surg} (right).

So we can always reduce case II to one where $\dg_1$ is has an L shape.
\end{itemize}



\begin{figure}
    \centering

    \begin{subfigure}[t]{0.3\textwidth}
        \centering
    \centering
        \includegraphics[width = 0.7\textwidth]{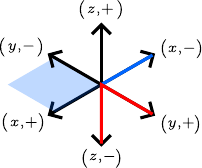}
    \caption{Solution for case II.B}
    \label{fig:D2=3,B}
    \end{subfigure}%
    ~ 
    \begin{subfigure}[t]{0.7\textwidth}
        \centering
    \includegraphics[width = 0.7\textwidth]{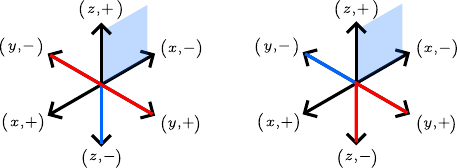}
    \caption{Solution for case II.C before (left) and after surgery (right).}
    \label{fig:D2=3,C_after_surg}

    \end{subfigure}
    \caption{}
\end{figure}

Case III: $|D(\dg_1)| =2$, $|D(\dg_2)| = 4$: In this case, we have $D(\dg_1) = \{(r_1,\sigma_1), (r_2,\sigma_2)\}$ and $D(\dg_2) = \{(r_3  ,\sigma_3), (r_4,\sigma_4), (r_5, \sigma_5), (r_6, \sigma_6)\}$. Solving out for GSC gives us 2 distinct solutions:
\begin{itemize}
    \item A: $r_2 = r_1$, $\sigma_2 = \overline{\sigma}_1$, $r_4 = r_3$, $\sigma_4 = \overline{\sigma}_3$, $r_6 = r_5$, $\sigma_6 = \overline{\sigma}_5$. Figure \ref{fig:D2=4} shows the solutions to this case.
    \item B: $r_3 = r_1$, $\sigma_3 = \overline{\sigma}_1$, $r_4 = r_2$, $\sigma_4 = \overline{\sigma}_2$, $r_6 = r_5$, $\sigma_6 = \overline{\sigma}_5$. There are two distinct solutions for this case corresponding to different sets $D_\pm(\dg_2)$, shown in Figure \ref{fig:D2=4}.
\end{itemize}
\begin{figure}
    \centering
    \includegraphics[width=0.8\textwidth]{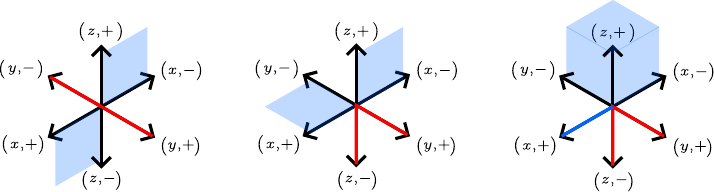}
    \caption{Solutions for case III.A (left) and III.B (center),(right).}
    \label{fig:D2=4}
    
\end{figure}
Case IV: $|D(\dg_1)| =3$, $|D(\dg_2)| = 3$: We of course have $D(\dg_1) = \{(r_1,\sigma_1), (r_2,\sigma_2)\}$ and $D(\dg_2) = \{(r_3  ,\sigma_3), (r_4,\sigma_4), (r_5, \sigma_5), (r_6, \sigma_6)\}$. Solving out for GSC gives us only a single solution:
\begin{itemize}
    \item A: $r_4 = r_1$, $\sigma_4 = \overline{\sigma}_1$, $r_3 = r_2$, $\sigma_3 = \overline{\sigma}_2$, $r_6 = r_5$, $\sigma_6 = \overline{\sigma}_5$. Figure \ref{fig:D1=3,D2=3} shows the possible solutions for this case. But notice that figure \ref{fig:D1=3,D2=3} (left) after surgery can be reduced to case IIIA, while \ref{fig:D1=3,D2=3} (right) can be reduced to case IIIB. So there are no distinct solutions in this case that haven't been covered in previous cases.
\end{itemize}
\begin{figure}
    \centering
    \includegraphics[width = 0.8\textwidth]{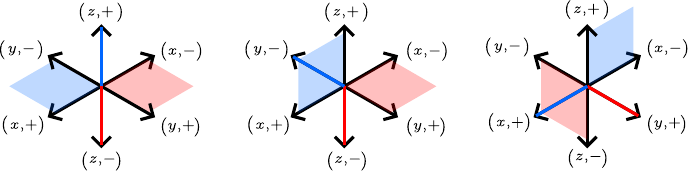}

    \caption{Solutions for case IV.A.}
    \label{fig:D1=3,D2=3}
\end{figure}

Note that we will also have discrete $\pi/2$ rotations and reflections about the $x,y,z$ axis of figures \ref{fig:D2} - \ref{fig:D1=3,D2=3} as solutions.

\subsection{Classification of 2 infinite flux strings}
We see from section \ref{sec:find_all_classes} that all possible cases of infinity directions are covered in cases I,II,III. Notice that in all cases, after surgery, simplifcation and relabelling, we have $|D(\dg_1)|=2$. We will divide the classification according to whether $D(\dg_1) =\{(r,\sigma), (r,\overline{\sigma})\} $ or $D(\dg_1) = \{(r,\sigma), (r',{\sigma}')\}$ with $r \neq r'$.\\

Let us first establish some useful defintions in the classification. Recall the definition \ref{def:halfpath} of a half-path $\dg_v^{(r,\sigma)}$ as having its start/end point as $v$ and $\dg_v^{(r,\sigma)}(t) \in \mathcal{E} ^{r,\sigma}$.\\



Let $p_{(r_1, \sigma_1, \tau_1), (r_2, \sigma_2, \tau_2)}$ denote a monotonic path such that
\begin{align}
    p_{(r_1, \sigma_1, \tau_1), (r_2, \sigma_2, \tau_2)}(t) \in \dg^{(r_1,\sigma_1)}_{\tau_1} \qquad t>t_+\\
    -p_{(r_1, \sigma_1, \tau_1), (r_2, \sigma_2, \tau_2)}(t) \in \dg^{(r_2,\sigma_2)}_{\tau_2} \qquad t<t_-
\end{align}
And let $\mathcal{P}_{(r_1, \sigma_1, \tau_1), (r_2, \sigma_2, \tau_2)}(t)$ denote the set of paths that are path equivalent to $p_{(r_1, \sigma_1, \tau_1), (r_2, \sigma_2, \tau_2)}(t)$.\\

Before the classification of the 2 string sectors we will need a few results:

\begin{proposition}
\label{thm:Pset}
If $\dg$ is a monotonic path with $D(\dg) = \{(r_1,\sigma_1), (r_2, \sigma_2)\}$, then we have $\dg \in \mathcal{P}_{(r_1, \sigma_1, \tau_1), (r_2, \sigma_2, \tau_2)}$
\end{proposition}
\begin{proof}
Since $\dg$ monotonic, $D_+(\dg) = (r_1, \sigma_1), D_-(\dg) = (r_2,\sigma_2)$ (upto relabelling of $D_\pm(\dg)$). Using lemma \ref{lem:region} we have a region $\dL$ such that $\dg(t) \in \mathcal{E}^{r_1,\sigma_1}$  and $-\dg(t) \in \mathcal{E}^{r_2,\sigma_2}$ for all $\dg(t) \in \mathcal{E}(\dL^c)$.\\

So there exist $t_\pm$ such that $\dg(t) \in \dg^{(r_1,\sigma_1)}_{\tau_1}$ for all $t>t_+$ and $-\dg(t) \in \dg^{(r_2,\sigma_2)}_{\tau_2}$ for all $t<t_-$ for particular $\tau_1,\tau_2$. This implies $\dg \in \mathcal{P}_{(r_1, \sigma_1, \tau_1), (r_2, \sigma_2, \tau_2)}$.
\end{proof}

We term $\mathcal{Q}^{D(\dg)}_\tau$ as the set of paths $\dg$ that have infinity directions $D(\dg)$ and there exists a half-path $\dg^{(r,\sigma)}_\tau \subset \dg$ for $(r,\sigma) \in D(\dg)$.\\

We also term $\mathcal{R}^{D(\dg)}$ as the set of paths $\dg$ that have infinity directions $D(\dg)$ with no additional restrictions.

\subsubsection{When $D(\dg_1) = \{(r,\sigma), (r,\overline{\sigma})\}$}
\label{subsec:2strcase1}
\begin{itemize}
    \item {$D(\dg_2) = \{(r_1, \sigma_1), (r_2, {\sigma}_2)\}$}:\\
    
Here GSC implies $r \neq r_1, r \neq r_2$. In this case the sectors are labelled by $$(g\in \mathbb{Z}_2, \dg_1 \in \mathcal{P}_{(r, \sigma, \tau), (r, \overline{\sigma}, \tau')}, \dg_2 \in \mathcal{P}_{(r_1, \sigma_1, \tau_1), (r_2, \sigma_2, \tau_2)})$$ $\dg_1,\dg_2$ are labelled using proposition \ref{thm:Pset}. We include $g \in \mathbb{Z}_2$ to take into account if the sector is charged or uncharged (which are the only two possibilities from the same reasoning as theorem \ref{thm:onesectorclass}). We will omit this reasoning for the proceeding cases. There is a special subcase when $r_1 = r_2$. Then GSC imposes $\sigma_1 = \overline{\sigma}_2$ and we have a slightly tighter constraint, $\dg_2 \in \mathcal{P}_{(r_1, \sigma_1, \tau_1), (r_1, \overline{\sigma}_1, \tau_2)}$


\item $D(\dg_2) = \{(r_1, \sigma_1), (r_2, {\sigma}_2), (r_3, \sigma_3), (r_4, {\sigma}_4)\}$\\

Here GSC implies $r \neq r_1 \neq r_2 \neq r_4 \neq r_3$, $r_2 = r_1, \sigma_2 = \overline{\sigma}_1$, $r_3 = r_4, \sigma_4 = \overline{\sigma}_3$. This simplifies to $D(\dg_2) = \{(r_1, \sigma_1), (r_1, \overline{\sigma}_1), (r_2, \sigma_2), (r_2, \overline{\sigma}_2)\}$. For this case we have, $$(g \in \mathbb{Z}_2, \dg_1 \in \mathcal{P}_{(r, \sigma, \tau), (r', \overline{\sigma}, \tau')}, \dg_2 \in \mathcal{R}^{D(\dg_2)})$$
Here $
\dg_2 \in \mathcal{R}^{D(\dg_2)}$ follows by definition, and the rest of the labelling is exactly the same as the previous cases.

\end{itemize}

\subsubsection{When $D(\dg_1) = \{(r,\sigma), (r',{\sigma}')\}$ with $r \neq r'$}
\label{subsec:2strcase2}

\begin{itemize}
    \item \label{item:qref} $D(\dg_2) = \{(r_1, \sigma_1), (r_2, {\sigma}_2), (r_3, \sigma_3)\}$:\\

    Here GSC implies $r \neq r_1 \neq r_2 \neq r_3$, $r_2 = r_1$, $\sigma_2 = \overline{\sigma}_1$. This simplifies as $D(\dg_2) = \{(r_1, \sigma_1), (r_1, \overline{\sigma}_1), (r_3, \sigma_3)\}$. In this case we have the labelling $$(g \in \mathbb{Z}_2, \dg_1 \in \mathcal{P}_{(r, \sigma, \tau), (r', \overline{\sigma}, \tau')}, \dg_2 \in \mathcal{Q}_{\tau_2}^{D(\dg_2)})$$
Of course $\dg_1$ labelling is given by proposition \ref{thm:Pset}. With the simplification of $D(\dg_2)$, it needs to be divided into $D_\pm(\dg_2)$. Since $D_\pm(\dg_2) \neq \emptyset$ we are left only with the possibility that $|D_+(\dg_2)| = 2$ and $|D_-(\dg_2)| = 1$ (upto relabelling of $D_\pm(\dg_2)$). This implies  $\dg^{(s,\rho)}_\tau \subset -\dg_2$ for $(s,\rho) \in D_-(\dg_2)$. So it follows $\dg_2 \in \mathcal{Q}_{\tau_2}^{D(\dg_2)}$.
\item $D(\dg_2) = \{(r_1,\sigma_1), (r_2, \sigma_2), (r_3, \sigma_3)\}$\\
GSC implies $r_1 = r, \sigma_1 = \overline{\sigma} , r_3 = r_2, \sigma_3 = \overline{\sigma}_2$. In this case we have the labelling of sector as $$(g \in \mathbb{Z}_2, \dg_1 \in \mathcal{P}_{(r, \sigma, \tau), (r', \sigma', \tau')}, \dg_2 \in \mathcal{Q}_{\tau''})$$.

\item $D(\dg_2) = \{(r_1,\sigma_1), (r_2, \sigma_2), (r_3, \sigma_3), (r_4, \sigma_4)\}$

GSC gives us $r_1 = r, \sigma_1 = \overline{\sigma}$,$ r_2 = r', \sigma_2 = \overline{\sigma}'$,$ r_4 = r_3, \sigma_4 = \overline{\sigma}_3$. This simplifies $D(\dg_2)$ to $$D(\dg_2) = \{(r,\overline{\sigma}), (r', \overline{\sigma}'), (r_3, \sigma_3), (r_3, \overline{\sigma}_3)\}$$ Now we have two possibilities: $|D_+(\dg_2)| = 3, |D_-(\dg_2)| = 1$ or $|D_+(\dg_2)| = 2, |D_-(\dg_2)| = 2$ (upto relabelling of course).

When $|D_+(\dg_2)| = 3, |D_-(\dg_2)| = 1$ our labelling of sector is $$(g\in \mathbb{Z}_2, \dg_1 \in \mathcal{P}_{(r, \sigma, \tau), (r', \sigma', \tau')}, \dg_2 \in \mathcal{Q}^{D(\dg_2)}_{\tau_2})$$.\\

When $|D_+(\dg_2)| = 2, |D_-(\dg_2)| = 2$ our labelling of sector is $$(g\in \mathbb{Z}_2, \dg_1 \in \mathcal{P}_{(r, \sigma, \tau), (r', \sigma', \tau')}, \dg_2 \in \mathcal{R}^{D(\dg_2)})$$ The labelling $\dg_2 \in \mathcal{R}^{D(\dg_2)}$ follows straightforwardly from the definition.
\end{itemize}

%% file: Main/GS_on_3str.tex
\section{Ground States on 3+ string configurations}
\label{sec:3str}
\subsection{3 string configurations}
The classification of states with 3 infinite flux strings proceeds similarly to that of 2 strings, but turns out to be much simpler.
\gss*

\begin{proof}
The proof proceeds in the same way as theorem \ref{thm:2monoGS}. If after performing surgery we find any $\dg'_n$ that's pathologically non-monotonic, then $\omega_{\dg_1,\cdots ,\dg_N}$ is not in a ground state sector.
\end{proof}

For the case of 3 infinite flux strings, we must have $|D(\dg_i)|=2$ for $i = 1,2,3$. Let $D(\dg_i) = \{(r_{2i-1},\sigma_{2i-1}), (r_{2i},\sigma_{2i})\}$. Then solving for GSC leads to 3 cases:
\begin{itemize}
    \item case A: $r_{2i-1} = r_{2i}$, $\sigma_{2i-1} = \overline{\sigma}_{2i}$
    \item case B: $r_3 = r_1$, $\sigma_3 = \overline{\sigma}_1$, $r_4 = r_2$, $\sigma_4 = \overline{\sigma}_2$, $r_6 = r_5$, $\sigma_6 = \overline{\sigma}_5$
    \item case C: $r_3 = r_1$, $\sigma_3 = \overline{\sigma}_1$, $r_5 = r_2$, $\sigma_5 = \overline{\sigma}_2$, $r_6 = r_4$, $\sigma_6 = \overline{\sigma}_4$
\end{itemize}

\begin{figure}
    \centering

    \includegraphics[width = 0.8\textwidth]{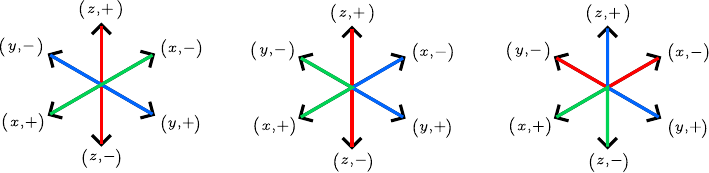}
    \caption{3 strings case A (left), B (middle), C(right)}
    \label{fig:3str}
\end{figure}

All 3 cases are shown in \ref{fig:3str}. Cases B,C can, after surgery and relabelling, be reduced to case A. 

\threesectorclass*
\begin{proof}
    We already know from solving for GSC that all paths $\dg_i$ can be simplified to the case of $D(\dg_i) = \{(r_i, \sigma_i), (r_i, \overline{\sigma}_i)\}$. Using proposition \ref{thm:Pset} we know $\dg_i \in \mathcal{P}_{(r_i, \sigma_i, \tau_i), (r_i, \overline{\sigma}_i, \tau'_i)}$. Once again, $g$ indicates whether the sector is charged or uncharged, using the same reasoning as theorem \ref{thm:mono-GS}. Here we have simply proven that $\hilb_{\{\dg_i\}}$ is a ground state sector. If $\omega_{\{\dg_i\}} \in \hilb_{\{\dg_i\}}$ is not a ground state, then using lemma \ref{lem:region} we can choose $\dL$ outside which $\dg_i(t) \in \cup_{(r,\sigma) \in D(\dg_i)} \mathcal{E}^{r,\sigma}$. Using surgery and straightening inside $\dL$ will lead to a new configuration of paths $\{\dg'_i\}$ and give us a ground state $\omega_{\{\dg'\}} \in \hilb_{\{\dg_i\}}$.
\end{proof}
\subsection{4+ string configurations}
\label{sec:4str}
Classification of $4+$ strings is the easiest.
\foursectorclass*
\begin{proof}
    From theorem \ref{thm:GSS}, we must have $\bigcap_{i=1}^N D(\dg_i) = \emptyset$. For a monotonic path $\dg_i$, $|D(\dg_i)|\ge2$. If $\bigcap_{i=1}^N D(\dg_i) = \emptyset$, then $|\bigcap_{i=1}^N D(\dg_i)| \ge 2N$. However for $\Gamma = \mathbb{Z}^3$, there only 6 unique Infinity directions given by $(r,\sigma)$ with $r \in \{x,y,z\}$ and $\sigma \in \{\pm\}$. So $\bigcap_{i=1}^N D(\dg_i) \neq \emptyset$ for $N\ge 4$. This concludes our proof. 
\end{proof}

%% file: Main/Translations.tex

%% file: Main/Discussion.tex
\section{Discussion and outlook}
We have found configurations containing infinite flux strings to be genuine infinite energy physical states of the 3d Toric Code Model on a $\mathbb{Z}^3$ lattice. These configurations cannot be obtained by the action on the translation invariant ground state $\omega$ of any operator belonging to $\cstar{}$, the algebra of quasi-local operators in the Toric Code model. However, they are stable ground states of the model.\\

In the bulk of this paper, we have established a necessary and sufficient criterion for a configuration with multiple infinite flux strings to belong to a ground state sector. The criterion is a twofold statement: (a) the infinite flux strings must individually be path equivalent to a monotonic infinite flux string, (b) they must not have any infinity directions in common. If either of these conditions is not met, we have devised a "straightening procedure" in a finite region to obtain an explicit new state with a lower energy. If after applying the straightening procedure a finite amount of times one obtains a ground state, then this configuration lies in a ground state sector. \\

We then classified the different ground state sectors for different infinite flux string configurations, and find in particular, that any configuration with more than 3 infinite flux strings cannot be a ground state.\\

It is easy to generalise the construction to arbitrary finite Abelian groups, and also for different types of infinite oriented lattices. It would be interesting to understand the infinite flux strings from the perspective of fusion 2-categories. While the fusion 2-category structure has been explored in the case of the 3d Toric Code Model on a 3-Torus \cite{Kong2020-an}, it remains an open question to define string-braiding and fusion for infinite flux strings. Since the infinite flux string is not "transportable" (in the sense that a translated infinite flux string configuration is not path-equivalent to the original infinite flux string), braiding or fusion of 2 infinite flux strings is a seemingly meaningless question. To resolve this difficulty will be the scope of our work in the near future.

%% file: Main/purity.tex
\section{Purity of the ground states of the 3d Toric Code}
\label{app:purity}
In this section we will prove the purity of all ground states of the form $\omega_{v_1, \cdots, v_n; \dg_1,\cdots ,\dg_m}$, working closely along the lines of \cite{Bols2023-qg}.
\subsection{$\omega$ is pure}
We first begin by using the string-net construction of this ground state. Here in order to avoid confusion between state vectors and states, we choose to denote states as $\psi$, vector states as $\ket{\psi}$, and $\langle \psi, \phi \rangle$ for inner product.\\

Define $\Pi_n^V$ as the set containing all vertices in $\mathcal{V}(\Gamma)$ inside the $n \times n \times n$ region centered at the origin. Let $\Pi_n^E$ be the region consisting of all edges with at least one boundary vertex in $\Pi_n^V$, $\Pi_n^F$ be the set of all $f$ with at least one vertex in $\Pi_n^V$, $\partial \Pi^F_n$ as the boundary edges of $\Pi^F_n$.\\

Let $\hilb_{{n}} := \otimes_{e \in \Pi_n^E \cup \partial \Pi^F_n} \hilb_e$. Let $B^n := \prod_{f \in \Pi^F_n} (1+B_f)/2$; $A^n := \prod_{v \in \Pi_n^V} (1+A_v)/2$ be projectors acting on $\Pi_n^E$. Define $\mathcal{W}_n := B^n \hilb_{{n}}$ and $\mathcal{V}_n := A^n B^n \hilb_{{n}}$. So a vector $\ket{\psi} \in \mathcal{V}_n$ satisfies $A_v \ket{\psi} = B_f \ket{\psi} = \ket{\psi}$ for all $v \in \Pi_n^V ,f \in \Pi_n^F$, while $\ket{\psi} \in \mathcal{W}_n$ satisfies $B_f \ket{\psi} = \ket{\psi}$ for all $f \in \Pi_n^F$.\\

Start with a product vector state $\ket{\Omega^0_n} \in \hilb_{{n}}$ with $\sigma^z_e \ket{\Omega^0_n} = \ket{\Omega^0_n}$ for all $e \in \Pi^E_n$. We then have $B_f {\ket{\Omega^0_n}} ={\ket{\Omega^0_n}}$ for all $f \in \Pi^F_n$, so $\ket{\Omega^0_n} \in \mathcal{W}_n$.\\

Let us define a higher dimensional string net state (also called surface net state) $\ket{\alpha}$ on $\Pi_n^V$. First, we consider the set $\overline{\mathcal{S}}$ as the set of dual (not necessarily connected) surfaces lying entirely in $\Pi_n^E$. Let $\ket{\alpha} \in \overline{\mathcal{S}}$. Then a surface net state $\ket{\alpha} \in \hilb_{{n}}$ is given as follows:
\begin{align*}
    \sigma^z_e \ket{\alpha} &= -\ket{\alpha} \qquad e \perp \alpha\\
    \sigma^z_e \ket{\alpha} &= \ket{\alpha} \qquad e \not \perp \alpha\\
    B_f \ket{\alpha} &= \ket{\alpha} \qquad f \in \Pi_n^F 
\end{align*}
Here $e \perp \alpha$ is understood to have the same meaning as in the definition of the flux string operator. The last condition is known as the "no flux condition". It forces $\alpha \in \overline{\mathcal{C}}$, the set of closed dual surfaces lying entirely in $\Pi_n^E$. $\ket{\alpha}$ can be prepared from the product vector state $\Omega^0_n$ in the following way:$\ket{\alpha} = \prod_{e \in \alpha} \sigma^x_e \ket{\Omega^0_n}$. We define the collection of surface-nets as $\pack := \{\ket{\alpha}\}$.\\


Let $\gauge$ be the set of gauge transformations on $\hilb_{{n}}$ generated by $A_v$ for all $v \in \Pi_n^V$. We notice,
\begin{align*}
    \ket{\alpha} &= \prod_{e \in \alpha} \sigma^x_e \ket{\Omega^0_n}\\ 
    &= \prod_{v \in V_\alpha} A_v \ket{\Omega^0_n}\\
    &=
    G_\alpha \ket{\Omega^0_n} \qquad \qquad G_\alpha := \prod_{v \in V_\alpha} A_v
\end{align*}
Where $V_\alpha$ is the set of vertices lying inside $\alpha$. Clearly, $G_\alpha \in \gauge$. Since $[A_v, B_f] = 0$, $[G_\alpha,B_f] = 0$ for all $f \in \Pi_n^F$ and $G_\alpha \in \gauge$. We thus have $\ket{\alpha} \in \mathcal{W}_n$.\\


\begin{lemma}
    For surface nets $\ket{\alpha} = G_\alpha \ket{\Omega^0_n}, \ket{\alpha'} = G_{\alpha'} \ket{\Omega^0_n}$ there exists a unique $G \in \gauge$ such that $\ket{\alpha} = G \ket{\alpha'}$, and $\langle{\alpha}, {\alpha'}\rangle = \delta_{G_{\alpha}, G_{\alpha'}}$
\end{lemma}
\begin{proof}
    First, notice that since $\gauge$ is generated by $A_v$ for $v \in \Pi_n$, $G_1, G_2 \in \gauge \implies G_1 G_2 =G_2 G_1 \in \gauge$. So $\gauge$ forms an abelian group, with $G^{-1}:= G^\dagger$. We can show existence of $G$ by noting that $\ket{\alpha} = G_\alpha G_{\alpha'}^\dagger \ket{\alpha'}$.\\

    To show uniqueness, we first see that if $\ket{\alpha} = G_1 \ket{\alpha'} = G_2 \ket{\alpha'}$ then $\ket{\alpha'} = G_1^\dagger G_2 \ket{\alpha'} = G \ket{\alpha'}$. So $G_{\alpha'} \ket{\Omega^0_n} = G G_{\alpha'} \ket{\Omega^0_n}$ implying $G = G_1^\dagger G_2 = 1$. \\

    We show $\langle{\alpha} ,{\alpha'}\rangle = \delta_{G_{\alpha}, G_{\alpha'}}$ by first observing that $\langle{\Omega^0_n}, {\Omega^0_n}\rangle = 1$ and $\langle{\Omega^0_n}, G {\Omega^0_n}\rangle = 0$ for $G \in \gauge$ such that $G \neq 1$. Then, $\langle{\alpha}, {\alpha'}\rangle = \langle{\Omega^0_n}, G_\alpha^\dagger G_{\alpha'} {\Omega^0_n}\rangle = \delta_{G_\alpha, G_{\alpha'}}$ thus completing the proof.
\end{proof}

\begin{remark}
    $|\gauge| = 2^{n^3}$ as it is generated by all $A_v \in \Pi_n$. 
\end{remark}

 Consider a surface net $\ket{\alpha} \in \hilb_{{n}}$ such that $\prod_{e \in \partial \Pi_n^F} \sigma^z_e \ket{\alpha} = \ket{\alpha}$. There can be many ways in which this condition can be satisfied. We call $b(\alpha) \subset \partial \Pi^F_n$ a boundary condition of $\alpha$ such that $\sigma^z_e \ket{\alpha} = - \ket{\alpha}$ for all edges $e \in b(\alpha)$, such that it satisfies the condition $\prod_{e \in \partial \Pi^F_n} \sigma^z_e \ket{\alpha} = \ket{\alpha}$. We now find it useful to divide $\pack$ into $\packbc := \{\ket{\alpha}\restriction_{\cstar{\Pi_n^E}} \in \pack | b(\alpha) = b \subset \Pi_n^E \} $.\\

Define $\mathcal{W}_n^b := P_b \mathcal{W}_{{n}}$ where $P_b$ projects to edges in the specific boundary configuration $b$ along $\regionf$.

\begin{lemma}
    $\Wbc = span \{\ket{\alpha}| \ket{\alpha} \in \packbc\}$. Moreover, $\mathcal{W}_{{n}} = \bigoplus_b \mathcal{W}_n^{b}$
\end{lemma}
\begin{proof}
First, we notice that $\mathcal{W}_{{n}} = span\{\ket{\alpha} | \ket{\alpha} \in \pack\}$, since if there exists a product vector state $ \ket{\psi} \notin \pack$ then there exists $f \in \Pi_n^F$ such that $B_f \ket{\psi} =0$ and so $\ket{\psi} \notin \mathcal{W}_n$. \\

Next, $\mathcal{W}^b_{{n}} \subset \mathcal{W}_{{n}}$. We also have $\sum_b P_b=1$, and $P_{b'} \packbc = \delta_{b,b'} \packbc$. So $\ket{\alpha} \in \packbc$ clearly lies in $\mathcal{W}_{n}^b$. If $b(\alpha_1) = b_1, b(\alpha_2) = b_2$ with $b_1 \neq b_2$, then we have $\bra{\alpha_1} \ket{\alpha_2} = 0$. So $\Wbc = span \{\ket{\alpha}| \ket{\alpha} \in \packbc\}$. This implies $\mathcal{W}_{{n}} = \sum_b P_b \Wbc = \bigoplus_b \mathcal{W}_n^{b}$.
\end{proof}
\begin{remark}
    The action of $\gauge[n]$ leaves $\Wbc$ invariant, since $\partial \Pi^F_n \cap \Pi^E_n = \emptyset$.
\end{remark} 

Let $\ket{\eta^b_n}$ be a vector given by 
\begin{align}
    \ket{\eta^b_n} := \frac{1}{|\packbc|^{1/2}} \sum_{\alpha \in \packbc} \ket{\alpha}
\end{align}

\begin{proposition}
    $\mathcal{V}_n = A^n \mathcal{W}_n$ is one dimensional and spanned by $\ket{\eta^b_n}$
\end{proposition}
\begin{proof}
    We first notice that $$A^n = \prod_{v \in \Pi^V_n} (1+A_v)/2 = 
\frac{1}{|\gauge|} \sum_{\{g_v\} \in \gauge} A_v^{g_v} = \frac{1}{|\gauge|} \sum_{U \in \gauge} U$$\\

Then for any $\ket{\alpha} \in \packbc$ we know that $\ket{\alpha} \in \mathcal{W}^n$. Since $\{\ket{\alpha}\}$ is a good basis for $\mathcal{W}^n$ we may just work with $\ket{\alpha}$. Applying $A^n$ on $\ket{\alpha}$ we get:
\begin{align}
    A^n \ket{\alpha} = \frac{1}{|\gauge|} \sum_{U \in \gauge} U \ket{\alpha} =  \frac{1}{|\gauge|} \sum_{{\alpha'} \in \packbc} \ket{\alpha'} = \frac{|\packbc|^{1/2}}{|\gauge|} \ket{\eta^b_n}
\end{align}
So $A^n \ket{\alpha} \propto \ket{\eta^b_n}$ and hence $\mathcal{V}^n = A^n \mathcal{W}^n$ is one-dimensional and spanned by $\ket{\eta^b_n}$. This concludes the proof.
\end{proof}

\begin{remark}
    $|\packbc|$ is independent of $b$, since there exists a bijective unitary map $U_b: \packbc \rightarrow \pack^0$ supported on the boundary $\partial \regionf$ given by $$U_b = \prod_{e\in \partial \regionf} {(\sigma^x_e)}^{(1 - \sigma^z_e)/2}$$ Here $0$ is the trivial boundary condition.
    \label{rem:bcind}
\end{remark}

\begin{lemma}
    If $O \in \cstar{\Pi_{n-1}}$ then we have that $$\eta^b_n(O) := \langle \eta^b_n, O \eta^b_n \rangle$$ is independent of $b$.
\end{lemma}
\begin{proof}
    Since there always exists a unique unitary $U_b$ supported on $\partial \regionf$ (see remark above) such that $\ket{\alpha} = U_b \ket{\beta}$ such that $\alpha \in \pack^0, \beta \in \packbc$, we then have: $$U_b \ket{\eta^0_n} = \ket{\eta^b_n}$$

    Since $U_b$ is supported on $\partial \regionf$ we have $[U, O] = 0$. Thus we have $$\eta^b_n(O) = \langle \eta^0_n, U_b^\dagger O U_b \eta^0_n \rangle = \eta^0_n(O)$$ 

    This means an operator $O \in \cstar{\Pi_{n-1}}$ does not see the boundary condition $b$. This concludes the proof. 
\end{proof}

\begin{corollary}
    $\eta^b_m \restriction_{\cstar{\Pi_n}} = \eta^0_m \restriction_{\cstar{\Pi_n}}$ for $m\ge n+1$
\end{corollary}

We now shorten the notation to say $\psi|_m := \psi \restriction_{\cstar{{\Pi^E_m}}}$ for any state $\psi$. 

\begin{lemma}
    Let $\omega$ be the frustration free translation invariant ground state of the 3dTC satisfying $\omega(A_v) = \omega(B_f) = 1$ for all $v \in \mathcal{V}(\Gamma),f \in \mathcal{F}(\Gamma)$. Then we have $\omega|_n = \eta^0_m|_n$ for any $n$ 
\end{lemma}
\begin{proof}
    We have the convex decomposition of $\omega$ into countably many pure states $\omega_i$: $\omega = \sum_i \lambda_i \omega_i$. Since $|\omega_i(A_v)| \le ||A_v|| = 1$ and $0 \le \lambda_i \le 1$ we have $\omega_i(A_v) = 1$ and similarly $\omega(B_f) = 1$ for all $v \in \Pi^V_n, f \in \Pi^F_n$. So for a GNS vector $\Omega_i$  of $\omega_i$ we have $A_v \Omega_i = B_f \Omega_i = \Omega_i$. This means $\Omega_i \in \mathcal{V}_n$. Since $\mathcal{V}_n$ is 1-dimensional, we have $\omega_i = \eta^0_m|_n$ for all $n$.
\end{proof}

We can now extend $\eta^0_n|_m$ to the whole observable algebra $\cstar{}$ and call it $\omega_n$.

\begin{lemma}
    The sequence of states $\omega_n$ converges in the weak$-^*$ topology to a state $\omega$. 
\end{lemma}

\begin{proof}
    To show convergence, Let $O \in \cstar{\Pi^E_n}$, then we have for all $\omega_m(O) = \omega_{n+1}(O)$ for all $m \ge n+1$. Since $n$ was arbitrarily chosen, the state $\omega_n$ converges for all local observables. Since the set of all local observables is dense in $\cstar{}$ we have that $\omega_n$ converges in the weak$-^*$ topology to a state $\omega$.
\end{proof}

\puregs*
\begin{proof}

    To show uniqueness, we consider another frustration free ground state $\omega'$. Since it agrees with $\omega$ on all local observables, it is the same state as $\omega$.\\

    To show frustration freeness and purity, we see that $\omega$ is the unique ground state with 0 energy. Hence it is frustration free. Since $\omega$ is a unique frustration free ground state, $\omega$ must be pure. \\

    To show translation invariance, consider $T_x$ being a translation operator. We have:

    \begin{align*}
    \omega(O^\dagger \delta_x (O)) :&= \lim_{\Lambda\rightarrow \mathcal{E}(\Gamma)} \omega (O^\dagger[T_x^\dagger H_\Lambda T_x, O])\\
    &=  \lim_{\Lambda\rightarrow \mathcal{E}(\Gamma)}\omega(O^\dagger [H_\Lambda, O])\\
    &= \omega (O^\dagger \delta (O)) \ge 0
    \end{align*}
    
    So the ground state of the translated Hamiltonian is the same as the ground state $\omega$ of the original Hamiltonian, which is a unique frustration free ground state. Thus $\omega$ is translation invariant. This concludes our proof.
\end{proof}

\purity*
\begin{proof}
    Let's begin with the trivial case when there are no charges or strings in the ground state configuration ($m=n=0$). In this case, theorem \ref{thm:puregs} tells us that $\omega$ is pure. We have $\alpha_{v_1, \cdots, v_n; \dg_1,\cdots ,\dg_m}:= \alpha_{v_1} \circ \cdots \circ \alpha_{v_n} \circ \alpha_{\dg_1} \circ \cdots \circ \alpha_{\dg_m}$ and all $\alpha_v, \alpha_\dg$ are automorphisms. The composition of automorphisms is an automorphism. Since automorphisms map pure states to pure states, we have $\omega_{v_1, \cdots, v_n; \dg_1,\cdots ,\dg_m} := \omega \circ \alpha_{v_1, \cdots, v_n; \dg_1,\cdots ,\dg_m}$ is a pure state.
\end{proof}

%% file: Main/Lattice.tex
\section{Lattice facts}
\label{app:lattice}

Here we list some miscellaneous but important lattice facts.

\begin{lemma}
\label{lem:aut_conv}
    Let $\alpha_{v;n}(O) := F_{\gamma_{v;n}}OF_{\gamma_{v;n}}$. $\alpha_{v;n}(O)$ converges strongly to $\alpha_v(O)$.
\end{lemma}
\begin{proof}
    Let $O \in \cstar{m}$ where $\cstar{m}$ is the restriction of $\cstar{}$ to an $m \times m \times m$ region in $\Gamma$. Now let $n'>n>m$. Then we have,
    \begin{align*}
        \alpha_{v;n'}(O) &= F_{\gamma_{v;n'}}OF_{\gamma_{v;n'}}\\
        &= F_{\gamma_{v;n}} F_{\gamma_{n;n'}}OF_{\gamma_{n;n'}}F_{\gamma_{v;n'}}\\
        &=F_{\gamma_{v;n}} OF_{\gamma_{v;n}}\\
        &= \alpha_{v,n}(O)
    \end{align*}
    So $\alpha_{v}(O)$ converges for all $n \ge m$ for any $O \in \cstar{m}$. Since $m$ was otherwise arbitrary, $\alpha_{v}(O)$ converges for any local operator $O \in \cstar{loc}$. $\cstar{loc}$ is dense in $\cstar{}$ so $\alpha_v(O)$ converges for all $O \in \cstar{}$.
\end{proof}

\begin{lemma}
    $\alpha_v$ is independent of the orientation of the path $\gamma_{v;n}$.
\end{lemma}
\begin{proof}
    We have the following property of the charge operators $F_{\gamma_{v;n}}^\dagger = F_{-\gamma_{v;n}}= F_{\gamma_{v;n}}$. Then we have for any $O \in \cstar{}$,

\begin{align*}
    \alpha_v ({O})&= \lim_{n\rightarrow \infty} F_{\gamma_{v;n}} {O} F_{\gamma_{v;n}} \qquad {O} \in \cstar{}\\
    &= \lim_{n\rightarrow \infty} F_{-\gamma_{v;n}} {O} F_{-\gamma_{v;n}}\\
    &= \alpha_{v}(O)
\end{align*}
\end{proof}

\begin{lemma}
\label{lem:aut_conv_2}
    Let $\alpha_{\dS_{\dg_n}}(O) := F_{\dS_{\dg_n}} O F_{\dS_{\dg_n}}$. $\alpha_{\dS_{\dg_n}}(O)$ strongly converges to $\alpha_\dg(O)$.
\end{lemma}
\begin{proof}
    Let $O \in \cstar{m}$ where $\cstar{m}$ is the restriction of $\cstar{}$ to an $m \times m \times m$ region in $\Gamma$. Now let $n'>n>m$. Then we have,
    \begin{align*}
        \alpha_{\dS_{\dg_{n'}}}(O) &= F_{\dS_{\dg_{n'}}}OF_{\dS_{\dg_{n'}}}\\
        &= F_{\dS_{\dg_n}} F_{\dS_{(\dg_{n'} - \dg_n)}}OF_{\dS_{\dg_{n'}}}F_{\dS_{(\dg_{n'} - \dg_{n})}}\\
        &=F_{\dS_{\dg_n}} OF_{\dS_{\dg_n}}\\
        &= \alpha_{\dS_{\dg_n}}(O)
    \end{align*}
    So $\alpha_{\dS_{\dg}}(O)$ converges for all $n \ge m$ for any $O \in \cstar{m}$. Since $m$ was otherwise arbitrary, $\alpha_{\dS_{\dg}}(O)$ converges for any local operator $O \in \cstar{loc}$. $\cstar{loc}$ is dense in $\cstar{}$ so $\alpha_{\dS_{\dg}}(O)$ converges for all $O \in \cstar{}$.
\end{proof}

\begin{lemma}
    $\alpha_\dg = \alpha_{-\dg}$ where $-\dg$ is the path with a reverse orientation.
    \label{lem:orient_independent}
\end{lemma}
\begin{proof}
    We have $F_{\dS_\dg}^\dagger = F_{\dS_{-\dg}} = F_{\dS_\dg}$ where $\dS_{-\dg}$ is the surface with reverse boundary orientation. So we have for any $O \in \cstar{}$,

\begin{align*}
    \alpha_\dg ({O})&= \lim_{n\rightarrow \infty} F_{\dS_{\dg_n}} {O} F_{\dS_{\dg_n}} \qquad {O} \in \cstar{}\\
    &= \lim_{n\rightarrow \infty} F_{\dS_{-\dg_n}} {O} F_{\dS_{-\dg_n}}\\
    &= \alpha_{-\dg}(O)
\end{align*}
\end{proof}

\begin{theorem}
    The set $(\partial \dS \bigcup_{n=1}^N \dg_n)\setminus \bigcup_{n=1}^N(\partial \dS \cap \dg_n)$ is a set of paths.
    \label{thm:set_of_paths}
\end{theorem}
\begin{proof}
    For simplicity we will only consider the case when $\partial \dS$ overlaps with each $\dg_n$ once and only once, and that the paths $\dg_n$ do not overlap with each other. These generalisations are straightforward. From lemma \ref{lem:orient_independent} we know that the orientation of any path $\dg$ is irrelevant for the physical objects $\alpha_{\dg}$. We may then freely change the orientation of the paths $\dg_n$. We first change the orientations of $\dg_n$ such that $\partial \dS$ has the reverse orientation to $\dg_n$ at its overlap with $\dg_n$ for all $n$.\\

    \begin{figure}
        \centering
        \includegraphics[  width = 0.5\textwidth]{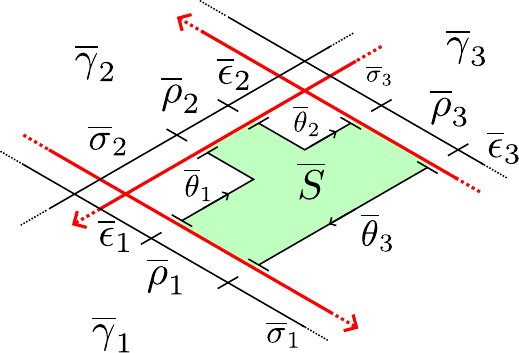}
        \caption{An example of division of $\dg_n, \dS$ into paths.}
        \label{fig:surgery}
    \end{figure}

    Call the set $\partial \dS \cap \dg_n$ as $\dr_n$. We now divide the paths $\dg_n$ into 3 paths each, starting from their beginning: $\de_n, \dr_n, \ds_n$. Here $\de_n$ consists of edges in $\dg_n$ before it overlaps with $\partial \dS$, $\dr$ is the set of all edges of $\dg_n$ that overlaps with $ \partial \dS$, and $\ds_n$ consists of all the edges after the overlap. Using the sets $\dr_n$ we can partition $\partial \dS$ into sets $\{\dt_i, \dr_i\}_{i=1}^{N}$ where $\dt_n$ is the set of edges of $\partial \dS$ between $\dr_n, \dr_{n+1}$. We illustrate this in figure \ref{fig:surgery} for the case of 3 paths. \\

    Let us now discard the orientation of the paths, as we are considering the paths purely as sets. Then we have
    \begin{align*}
        \partial \dS \bigcup_{n=1}^N \dg_n  &= \sum_n \de_n + \dr_n + \ds_n + \dt_n\\
\bigcup_{n=1}^N(\partial \dS \cap \dg_n) &= \sum_n \dr_n\\
\implies (\partial \dS \bigcup_{n=1}^N \dg_n)\setminus \bigcup_{n=1}^N(\partial \dS \cap \dg_n) &= \sum_n \de_n + \ds_n + \dt_n\\
&=\sum_n \de_n + \dt_n + \ds_{n+1\text{ mod }N}\\
&= \sum_n \dg_n'
    \end{align*}
    Where $\dg_n' = \de_n + \dt_n + \ds_{n+1\text{ mod }N}$.\\

    To show that $\dg_n'$ are paths, we restore the orientations and notice that $\dt_n$ is a path from $\partial \de_n$ to $\partial \ds_{n+1\text{ mod }N}$, where $\partial \de_n, \partial \ds_n$ are the finite end and start vertices of paths $\de_n, \ds_n$ respectively. So $\dg_n'$ has a consistent orientation. This concludes our proof.
\end{proof}

%% file: paperpile.bib
@ARTICLE{Nayak2008-ef,
  title     = "{Non-Abelian} anyons and topological quantum computation",
  author    = "Nayak, Chetan and Simon, Steven H and Stern, Ady and Freedman,
               Michael and Das Sarma, Sankar",
  journal   = "Rev. Mod. Phys.",
  publisher = "American Physical Society",
  volume    =  80,
  number    =  3,
  pages     = "1083--1159",
  month     =  sep,
  year      =  2008
}

@ARTICLE{Naaijkens2013-ji,
  title         = "Quantum spin systems on infinite lattices",
  author        = "Naaijkens, Pieter",
  month         =  nov,
  year          =  2013,
  archivePrefix = "arXiv",
  primaryClass  = "math-ph",
  eprint        = "1311.2717"
}

@ARTICLE{Bachmann2023-qv,
  title     = "Dynamical Abelian anyons with bound states and scattering states",
  author    = "Bachmann, Sven and Nachtergaele, Bruno and Vadnerkar, Siddharth",
  journal   = "J. Math. Phys.",
  publisher = "AIP Publishing",
  volume    =  64,
  number    =  7,
  month     =  jul,
  year      =  2023,
  language  = "en"
}

@ARTICLE{Kitaev2003-qr,
  title    = "Fault-tolerant quantum computation by anyons",
  author   = "Kitaev, A Yu",
  journal  = "Ann. Phys.",
  volume   =  303,
  number   =  1,
  pages    = "2--30",
  month    =  jan,
  year     =  2003
}

@ARTICLE{Alicki2007-pk,
  title     = "{A statistical mechanics view on Kitaev's proposal for quantum
               memories}",
  author    = "Alicki, R and Fannes, M and Horodecki, M",
  journal   = "J. Phys. A: Math. Theor.",
  publisher = "IOP Publishing",
  volume    =  40,
  number    =  24,
  pages     = "6451",
  month     =  may,
  year      =  2007,
  language  = "en"
}

@ARTICLE{Nakamura2020-vb,
  title     = "Direct observation of anyonic braiding statistics",
  author    = "Nakamura, J and Liang, S and Gardner, G C and Manfra, M J",
  journal   = "Nat. Phys.",
  publisher = "Nature Publishing Group",
  volume    =  16,
  number    =  9,
  pages     = "931--936",
  month     =  sep,
  year      =  2020,
  language  = "en"
}

@ARTICLE{Kong2020-an,
  title     = "Defects in the 3-dimensional toric code model form a braided
               fusion 2-category",
  author    = "Kong, Liang and Tian, Yin and Zhang, Zhi-Hao",
  journal   = "J. High Energy Phys.",
  publisher = "Springer Science and Business Media LLC",
  volume    =  2020,
  number    =  12,
  month     =  dec,
  year      =  2020,
  copyright = "https://creativecommons.org/licenses/by/4.0",
  language  = "en"
}

@ARTICLE{Cha2020-rz,
  title     = "On the stability of charges in infinite quantum spin systems",
  author    = "Cha, Matthew and Naaijkens, Pieter and Nachtergaele, Bruno",
  journal   = "Commun. Math. Phys.",
  publisher = "Springer Science and Business Media LLC",
  volume    =  373,
  number    =  1,
  pages     = "219--264",
  month     =  jan,
  year      =  2020,
  language  = "en"
}

@ARTICLE{Doplicher1971-jd,
  title     = "{Local observables and particle statistics I}",
  author    = "Doplicher, S and Haag, R and Roberts, J E",
  journal   = "Commun. Math. Phys.",
  publisher = "Springer",
  volume    =  23,
  pages     = "199--230",
  year      =  1971
}

@ARTICLE{Fiedler2015-na,
  title    = "{Haag duality for Kitaev's quantum double model for abelian
              groups}",
  author   = "Fiedler, L and Naaijkens, P",
  journal  = "Rev. Math. Phys.",
  volume   =  27,
  number   =  09,
  year     =  2015
}

@ARTICLE{Cha2018-ke,
  title     = "{The complete set of infinite volume ground states for Kitaev's
               abelian quantum double models}",
  author    = "Cha, M and Naaijkens, P and Nachtergaele, B",
  journal   = "Commun. Math. Phys.",
  publisher = "Springer Science and Business Media LLC",
  volume    =  357,
  number    =  1,
  pages     = "125--157",
  year      =  2018,
  language  = "en"
}

@ARTICLE{Bachmann2016-qx,
  title         = "Local disorder, topological ground state degeneracy and
                   entanglement entropy, and discrete anyons",
  author        = "Bachmann, Sven",
  month         =  aug,
  year          =  2016,
  archivePrefix = "arXiv",
  primaryClass  = "math-ph",
  eprint        = "1608.03903"
}

@ARTICLE{Stormer1999-bk,
  title    = "The fractional quantum Hall effect",
  author   = "Stormer, Horst L and Tsui, Daniel C and Gossard, Arthur C",
  journal  = "Reviews of Modern Physics",
  month    =  mar,
  year     =  1999
}

@ARTICLE{Von_Keyserlingk2013-zv,
  title     = "Three-dimensional topological lattice models with surface anyons",
  author    = "von Keyserlingk, C W and Burnell, F J and Simon, S H",
  journal   = "Phys. Rev. B Condens. Matter",
  publisher = "American Physical Society",
  volume    =  87,
  number    =  4,
  pages     = "045107",
  month     =  jan,
  year      =  2013
}

@BOOK{Bratteli2012-gd,
  title     = "Operator Algebras and Quantum Statistical Mechanics: Volume 1:
               C*- and W*- Algebras. Symmetry Groups. Decomposition of States",
  author    = "Bratteli, Ola and Robinson, Derek William",
  publisher = "Springer Science \& Business Media",
  month     =  dec,
  year      =  2012,
  language  = "en"
}

@ARTICLE{Naaijkens2010-aq,
  title         = "Localized endomorphisms in Kitaev's toric code on the plane",
  author        = "Naaijkens, Pieter",
  month         =  dec,
  year          =  2010,
  archivePrefix = "arXiv",
  primaryClass  = "math-ph",
  eprint        = "1012.3857"
}

@ARTICLE{Doplicher1974-hb,
  title    = "Local observables and particle statistics {II}",
  author   = "Doplicher, Sergio and Haag, Rudolf and Roberts, John E",
  journal  = "Commun. Math. Phys.",
  volume   =  35,
  pages    = "49--85",
  month    =  mar,
  year     =  1974
}

@ARTICLE{Bols2023-qg,
  title         = "The double semion state in infinite volume",
  author        = "Bols, Alex and Kjaer, Boris and Moon, Alvin",
  month         =  jun,
  year          =  2023,
  archivePrefix = "arXiv",
  primaryClass  = "math-ph",
  eprint        = "2306.13762"
}

@ARTICLE{Castelnovo2008-yg,
  title     = "Topological order in a three-dimensional toric code at finite
               temperature",
  author    = "Castelnovo, Claudio and Chamon, Claudio",
  journal   = "Phys. Rev. B Condens. Matter",
  publisher = "American Physical Society",
  volume    =  78,
  number    =  15,
  pages     = "155120",
  month     =  oct,
  year      =  2008
}

@ARTICLE{Kulkarni2019-gw,
  title     = "Decoding the three-dimensional toric codes and welded codes on
               cubic lattices",
  author    = "Kulkarni, Abhishek and Sarvepalli, Pradeep Kiran",
  journal   = "Phys. Rev. A",
  publisher = "American Physical Society",
  volume    =  100,
  number    =  1,
  pages     = "012311",
  month     =  jul,
  year      =  2019
}
